\documentclass[lettersize,journal]{IEEEtran}
\usepackage{amsmath,amssymb,mathtools,amsthm,amsfonts,array}
\usepackage{subfigure,graphicx,graphics,color,psfrag}
\usepackage{cite}
\usepackage{balance,caption}
\allowdisplaybreaks

\usepackage{algorithm,algorithmicx, algpseudocode}
\usepackage{accents,bm,url}
\usepackage[english]{babel}
\usepackage{multirow,enumerate,cases}
\usepackage{stfloats,dsfont,soul,fancyhdr,hhline}

\newtheorem{definition}{Definition}
\newtheorem{theorem}{Theorem}
\newtheorem{lemma}{Lemma}
\newtheorem{remark}{Remark}

\newtheorem{proposition}{Proposition}

\begin{document}

\title{Geometric View on Integrated Cascaded Channel of IRS-aided Communications}

\author{Yunli Li and 
    Young Jin Chun, \IEEEmembership{Member, IEEE}% <-this % stops a space
    \thanks{This work was supported in part by the Early Career Scheme under Project 21205021 and the General Research Fund under Project 11211122, both established under the University Grant Committee of the Hong Kong Special Administrative Region, China; in part by the City University of Hong Kong (CityU) under Project 7020083, and Project 7006090.}
    \thanks{Y.~Li and Y.~J.~Chun are with the Department of Electrical Engineering, City University of Hong Kong, Hong Kong, China. (e-mail: yunlili2-c@my.cityu.edu.hk; yjchun@cityu.edu.hk)}}
 
% The paper headers
\markboth{Journal of \LaTeX\ Class Files,~Vol.~14, No.~8, August~2021}%
{Shell \MakeLowercase{\textit{et al.}}: A Sample Article Using IEEEtran.cls for IEEE Journals}

\markboth{}
{Shell \MakeLowercase{\textit{et al.}}: Bare Demo of IEEEtran.cls for IEEE Journals}

\maketitle

\begin{abstract}

The hybrid intelligent reflecting surface (IRS) architecture is a novel technology that leverages the advantages of both passive and active IRS;  the passive IRS offers a large aperture, while the active IRS provides additional power amplification. Prior studies have shown that the optimal performance of IRS-assisted wireless networks is achieved when the passive IRS is deployed near the transceivers and the active IRS is closer to the receiver with decreasing amplification power, assuming that the IRS is located between transceivers and that the height difference between IRS and transceivers is much smaller than the link distance between transceivers. However, most of the prior works on IRS blindly adopted this assumption in the IRS association policy, which essentially becomes a partial selection strategy that offers analytical simplicity at the cost of sub-optimal performance. This limitation motivated us to find the globally optimal deployment strategy for all types of IRS. To this end, we first employ the \emph{geometric models} for integrated path loss distance (known as Cassini oval and Ellipse for product- and sum-distance path loss laws, respectively) and use them to determine the optimal locations of the hybrid IRS. Then, we design a novel \emph{opportunistic association policy} for hybrid IRS based on the integrated path loss model. Furthermore, we validate our proposed methods through simulations and show that they significantly outperform the conventional nearest association policy, especially for hybrid and active IRS.

\end{abstract}


\IEEEpeerreviewmaketitle

\section{Introduction}
The Intelligent Reflecting Surface (IRS), also known as Reconfigurable Intelligent Surface (RIS) and Large-Scale Intelligent Surface (LIS), is considered as a critical technology capable of dynamically altering signal paths, extending coverage, and simultaneously enhancing network capacity for future wireless networks \cite{wu2021intelligent, di2020smart, wu2019towards, gong2020toward, elmossallamy2020reconfigurable}. There are three types of IRS architecture: passive IRS, active IRS, and hybrid IRS. 

The passive IRS consists of a phase shift controller and a large surface equipped with numerous reflecting elements. The passive reflecting elements merely reflect signals without additional active processing, which enables high spectral- and energy-efficient communication at a low cost \cite{wu2021intelligent,8989805,lyu2021hybrid}. Despite the performance gain provided by passive IRS, the enhancement can be significantly impeded by the severe two-fold path loss. Specifically, the benefit of passive IRS is limited when the direct link between transceivers is stronger than the IRS-aided links \cite{wu2019intelligent,10458985}.

Active IRS can overcome the severe path loss by simultaneously reflecting and amplifying the received signals \cite{long2021active,zhang2022active,zhi2022active,liu2021active,khoshafa2021active,9896755}. The integration of amplifiers into the active reflecting elements can be achieved through various existing active components, such as current-inverting converters \cite{lonvcar2019ultrathin}, asymmetric current mirrors \cite{6047578} or integrated circuits \cite{kishor2011amplifying}. Nevertheless, unlike passive IRS that operates in a noise-free mode, the aperture of active IRS is constrained by the amplification power and thermal noise generated by active parts. 

To combine the benefits of passive and active IRS, a novel hybrid IRS architecture is proposed in \cite{kang2023active}, which contains two sub-surfaces: one of the sub-surfaces is composed of passive reflecting elements, while the other sub-surface consists of active reflecting elements and additional operation power is required.\footnote{The definition of hybrid IRS is different from \cite{schroeder2022two}, where the active reflecting element are connected an RF chain to support channel estimation and without amplification.} The hybrid IRS leverages the advantages of both passive and active IRS in a complementary manner. Specifically, the passive IRS offers asymptotic beamforming gain with a squared power scaling order of $O(N^{2})$, while the active IRS provides an extra power amplification gain but with a linear power scaling order of $O(N)$ for beamforming due to thermal noise introduced by the active parts, where $N$ is the number of reflecting elements \cite{you2021wireless}.

However, implementing the innovative hybrid architecture presents a difficult issue due to disparate location preferences for passive and active sub-surfaces. Specifically, the passive sub-surface is inclined towards being situated in close proximity to the transceivers, while the active sub-surface necessitates placement at a certain distance away from the transmitter to ensure an acceptable amplification factor \cite{wu2019intelligent,you2021wireless}. Furthermore, active sub-surface is supposed to be deployed closer to the receiver with decreasing amplification power \cite{you2021wireless}. Consequently, finding the optimal deployment strategy is one of the most essential design challenges for the hybrid IRS-assisted wireless networks to achieve the full potential of both passive and active sub-surfaces.

The channel properties are crucial for analyzing the end-to-end performance of IRS-aided communications. As summarized in \cite{huang2022reconfigurable}, there are numerous small-scale multipath fading models, including the phase shift model, transmission-mode model, physical-based channel model, beamspace channel model, geometry-based stochastic model, keyhole channel model, and machine learning-based channel model. Rician fading model, especially, is popular in IRS-assisted communication due to its line-of-sight assumption \cite{al2021performance,elhattab2022ris}. In \cite{alayasra2021irs}, the authors proposed a beamspace model based on extended Saleh-Valenzuela (SV) channel, categorizing the scatters and paths of multipath components. Mixture Gamma distribution-based, generalized fading model was introduced in \cite{10458985} to ease the analysis of IRS-aided networks.
		
To the best of our knowledge, the geometric modeling of the integrated path loss distance remains an open problem, although many studies have modeled small-scale fading in the cascaded channel. To unlock the full potential of hybrid IRS, it is crucial to comprehend the geometric arrangement of node locations for optimal IRS placement. Furthermore, we need to develop an opportunistic association policy for IRS, which associates IRS based on the best end-to-end performance \cite{liu2003framework}.

\subsection{Geometric Model}

The geometric model offers a realistic representation of wireless networks and provides a spatial understanding of the given wireless network. The location of nodes affects signal strength, interference, and ultimately the performance of the network. The geometric model can accurately portray the physical distances and relative positions of nodes, allowing for better predictability of system behavior \cite{tse2005fundamentals}. 

Comprehending the geometric nature of a network enables the anticipation of signal propagation, and identification of interference patterns and potential bottlenecks. This understanding proves invaluable in optimizing the deployment and node density to enhance overall network performance \cite{tse2005fundamentals}. For example, in machine-to-machine (M2M) communications, where the fixed locations of user equipment (UE) allow for the strategic design of the deployment of IRSs based on defined objective functions, such as maximizing IRS utility or minimizing interference among UEs. The geometric model further plays a pivotal role in informing decisions regarding the integration of a new node. For instance, once the system's deployment is established, UEs can access a look-up table of potential average performance for each IRS from the base station (BS) upon joining the cell and reporting their locations. Subsequently, UE can judiciously select the serving IRS based on the information provided in the table.

Furthermore, the geometric model exhibits adaptability across diverse wireless scenarios, encompassing both mobile and static networks, while seamlessly accommodating variations in network size and topology. Notably, in IRS-aided networks, geometric model is of critical importance in enhancing the analytical framework by precisely capturing the statistics of the integrated links based on practical opportunistic association policies, as opposed to relying on partial selection over the nearest association. Moreover, the geometric model serves as a robust foundation for stochastic geometry-based system-level analysis. This is attributed to their ability to model the distribution of connection distances among nodes based on their inherent geometric relationships \cite{andrews2016primer,andrews2011tractable,haenggi2012stochastic}.

The ellipse model, extensively explored for sum-distance, has gained much attention in localization problems. However, to the best of the authors' knowledge, this is the first exploration of the Cassini oval model for product-distance in the realm of wireless communications.

\subsection{Association Policy: Nearest and Opportunistic}

For IRS-aided communication, accurate statistical modeling of the integrated cascaded channel remains an open problem due to the lack of exploration in the integrated path loss distance. Most existing literature relies on the nearest association policy, which redirects the focus to individual link by trading accuracy for tractability \cite{8989805,lyu2021hybrid,10458985}. 

The nearest association policy entails a partial selection mechanism for passive IRS grounded by a simple setup with the product-distance path loss law, where the location of the IRS is confined to the segment between the BS and UE with limited height \cite{you2022deploy,wu2021intelligent}. Additionally, this association relies on a widely adopted common assumption for analytical simplification, where the link distance between the associated passive IRS and the serving BS of the typical UE, referred to as the BS$\rightarrow$IRS link, equals the link distance between the BS and UE, denoted as BS$\rightarrow$UE link. This simplification directs the analysis toward the link distance between the IRS and UE, denoted as IRS$\rightarrow$UE link, rather than considering the integrated cascaded BS$\rightarrow$IRS$\rightarrow$UE channel. 

% why unsuitable for passive
However, this assumption proves impractical for several reasons. First, the approximated nearest association is effective solely for the product-distance path loss law. Yet, diverse propagation scenarios necessitate distinct path loss laws, as corroborated by \cite{tang2020wireless}. For example, the product-distance path loss law is appropriate for far-field and near-field beamforming cases, whereas the sum-distance path loss law is suitable for near-field broadcasting cases. Regrettably, the nearest association policy is ill-suited for the sum-distance path loss law. Second, in IRS-aided wireless networks, the channel condition of \emph{the integrated cascaded channel}, i.e., BS$\rightarrow$IRS$\rightarrow$UE, directly governs end-to-end communication quality, instead of the individual links, i.e., BS$\rightarrow$IRS or IRS$\rightarrow$UE. Specifically, even if one of the single links attains sufficient quality, poor performance may persist if the other link is obstructed.

Moreover, the nearest association policy fails to realize full potential of the active and hybrid IRS. As illustrated in \cite{you2021wireless}, the active IRS should be deployed closer to the receiver with decreasing amplification power to maximize the achievable rate. In addition, for hybrid IRS, the deployment preferences of the two sub-surfaces, namely passive and active, along the segment between the BS and UE are different. As such, within this constrained deployment framework along the segment between BS and UE, a performance balance arises between these two sub-surfaces regarding to IRS locations. This conflict deployment preference between the two sub-surfaces has motivated us to investigate the optimal locations and opportunistic association of hybrid IRS across the entire area, which requires comprehensive understanding of the integrated channel.

The integrated cascaded channel imposes elevated demands on the accurate modeling of communication, necessitating a delicate equilibrium between analytical complexity and accuracy. Our exploration leverages the geometric properties of the integrated path loss distance within the IRS-aided cascaded channel. While previous work, as outlined in \cite{fang2022optimum}, has delved into the optimal selection of passive IRS in a single-cell network using stochastic geometry tools, further research into the network’s geometric nature is imperative for studying the optimal locations of the hybrid IRS. Specifically, we need to investigate the cascaded channel’s end-to-end performance from a geometric perspective, which is crucial for developing an opportunistic association policy and a deployment strategy tailored for hybrid IRS.

\subsection{Contributions}

Driven by the aforementioned concerns, our central objective is to address the \emph{geometric modeling} intricacies inherent in the equal-gain integrated path loss distance within the IRS-assisted cascaded channel. Subsequently, we formulate an optimization problem to ascertain the optimal locations for the hybrid IRS and put forth an \emph{opportunistic association policy} tailored for hybrid IRS. The main contributions of this paper may be summarized as follows:
 
\begin{enumerate}
	\item We examine two geometric models of the integrated path loss distance for an IRS-aided cascaded channel, referred to as Cassini oval and Ellipse for product- and sum-distance path loss laws, respectively. Building upon the geometric models, we explore the statistical characteristics of the integrated path loss distance, such as equal-gain trajectory, cumulative distribution function (CDF) and probability density function (PDF).  
	\item Based on the geometric nature of node locations, we demonstrate that the optimal locations for the hybrid IRS can be determined by the integrated path loss distance of the cascaded channel (BS$\rightarrow$IRS$\rightarrow$UE), system parameters, and the link distance between the BS and UE. Moreover, the outcomes for both passive and active IRS can be obtained as special instances of the hybrid IRS.
	\item We investigate potential implementations of the proposed geometric models in IRS-aided networks and compare the performance of opportunistic and nearest association policies for IRS. It is noted that the nearest association policy is ineffective for hybrid/active IRS, while the opportunistic association policy results in significant performance improvement across all types of IRS.
\end{enumerate}

\subsection{Organizations}
The remaining paper is organized as follows. Section \ref{sec:sysM} introduces the system model, Section \ref{sec:geo} presents the geometric models and derives statistical characteristics for the equal-gain integrated path loss distance of the cascaded channel. Section \ref{sec:assoc} displays the channel statistics in IRS-aided communication, establishes the hybrid IRS's optimal locations and proposes an opportunistic association policy. Section \ref{sec:specialCase} analyzes special cases to evaluate our proposed geometric model. Section \ref{sec:numericalR} provides Monte-Carlo simulations to verify our analysis and Section \ref{sec:conclusion} concludes the paper.

\section{System Model}\label{sec:sysM}

\begin{figure}[t!]
	\centering
	\includegraphics[width=\linewidth]{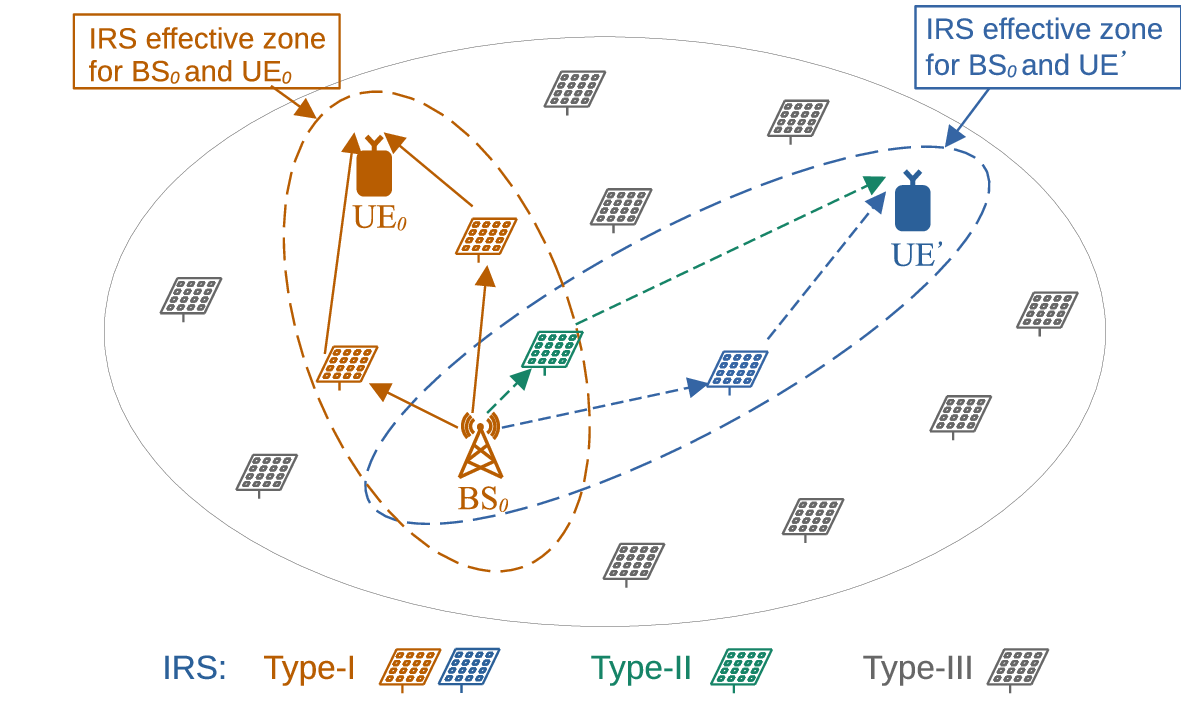}
	\caption{Downlink of a hybrid IRS-aided network. The configuration of the IRS effective zone for transceiver pairs undergoes variations based on the ratio of the integrated path loss distance of the cascaded channel to the link distance between the BS and UE. In accordance with the IRS effective zone, the IRSs are categorized into three types: Type-I, the IRSs located within the IRS effective zone for a single transceiver pair; Type-II, the IRSs positioned within the intersection of the IRS effective zones for more than one transceiver pairs; Type-III, the IRSs situated outside the IRS effective zones. }
	\label{fig:multipleCellNetwork}
\end{figure}

We consider a hybrid IRS-aided network, as illustrated in Fig. \ref{fig:multipleCellNetwork}, where both BSs and UEs are equipped with a single antenna and orthogonal multiple access is employed within each cell. Without loss of generality, we assume that each BS has an infinitely backlogged queue, i.e., the BS always has data to transmit to each user, \cite{wang2014increasing,liu2003framework}. Each hybrid IRS has $N$ reflecting elements and contains two sub-surfaces. One of the sub-surfaces contains $N_{\rm pas}$ passive reflecting elements along with a phase-shift controller, while the other sub-surface contains $N_{\rm act}$ active reflecting elements, a phase-shift controller, and an additional amplify controller, as shown in Fig. \ref{fig:hybridIrs} \cite{kang2023active}. We assume that $J$ hybrid IRSs are uniformly distributed throughout the whole area, and we denote the set of IRSs as $\mathcal{J}\triangleq\{{\rm IRS}_{1},\cdots,{\rm IRS}_{J}\}$. 

The reflection matrix serves as a crucial distinguishing factor between passive and active sub-surfaces. Specifically, the reflection matrix of the passive sub-surface is presented as ${\bf\Psi}_{\rm pas}\triangleq{\rm diag}(e^{j\phi_{1}^{\rm pas}},\cdots,e^{j\phi_{N_{\rm pas}}^{\rm pas}})\in \mathbb{C}^{N_{\rm pas}\times N_{\rm pas}}$, where $\phi_{n_{\rm pas}}^{\rm pas}$ represents the phase shift of the $n_{\rm pas}$-th passive reflecting element with $n_{\rm pas}\in\mathcal{N}_{\rm pas}\triangleq\{1,\cdots,N_{\rm pas}\}$, and the reflection amplitude of the passive reflecting elements are set as one. In contrast, the reflection matrix of the active sub-surface is denoted as ${\bf\Psi}^{({\bf \eta})}_{\rm act}\triangleq {\rm diag}(\eta_{1}e^{j\phi_{1}^{\rm act}},\cdots, \eta_{N_{\rm act}}e^{j\phi_{N_{\rm act}}^{\rm act}})\in \mathbb{C}^{N_{\rm act}\times N_{\rm act}}$, where $\phi_{n_{\rm act}}^{\rm act}$ represents the phase shift of the $n_{\rm act}$-th active reflecting element with $n_{\rm act}\in\mathcal{N}_{\rm act}\triangleq\{1,\cdots,N_{\rm act}\}$, $\eta_{n_{\rm act}}$ denotes the reflection amplitude of the $n_{\rm act}$-th active reflecting element. Without loss of generality, for the associated IRS, we assume a common amplification factor across all active reflecting elements,\footnote{As shown in \cite{kang2023active}, all active reflecting elements are supposed to adopt identical amplification factor as they suffer the same path loss. Nonetheless, this may not suitable for multi-UE scenario as the beamforming design should balance the tradeoff among numerous UEs.} i.e., $\eta_{n_{\rm act}}^{2}\triangleq \frac{P_{\rm F}}{N_{\rm act}(P_{\rm T}\mathbb{E}[G_{\rm BI}]\epsilon d_{\rm BI}^{-\alpha}+\delta_{\rm F}^{2})},\forall n_{\rm act}$ \cite{lonvcar2020challenges,amato2018tunneling,you2021wireless}, where $P_{\rm F}$ is the amplification power of the active reflecting elements, $\mathbb{E}[G_{\rm BI}]$ is the average small-scale channel gain of BS$\rightarrow$IRS link, $\delta_{\rm F}^{2}$ is the thermal noise power generated by the active reflecting elements, $P_{\rm T}$ is the transmit power of the BS,  $\epsilon$ represents the reference channel power gain at a distance of 1 meter (m), and $\alpha$ is the path loss exponent ($\alpha>2$) \cite{you2021wireless}. Then, the reflection matrix of the active sub-surface can be further expressed as ${\bf\Psi}^{({\bf \eta})}_{\rm act}\triangleq\eta{\bf\Psi}_{\rm act}$, where $ {\bf\Psi}_{\rm act}= {\rm diag}(e^{j\phi_{1}^{\rm act}},\cdots,e^{j\phi_{N_{\rm act}}^{\rm act}})\in \mathbb{C}^{N_{\rm act}\times N_{\rm act}}$.

We examine the cascaded channel between a typical UE and its serving BS via the associated IRS. We present an example case in which the signal reflected by all reflecting elements on the same IRS experiences identical fading and path loss. It is noted that the quality of communication for the cascaded channel is determined by the path loss and the small-scale fading of the integrated cascaded channel (BS$\rightarrow$IRS$\rightarrow$UE) rather than those of each individual link. In other words, even if one link of the cascaded channel is lossless, severe blockage of the other link may still render the entire cascaded channel non-functional. Therefore, this work emphasizes the integrated channel condition of the cascaded BS$\rightarrow$IRS$\rightarrow$UE channel instead of focusing on individual link conditions, i.e., BS$\rightarrow$IRS link or IRS$\rightarrow$UE link. To facilitate subsequent analysis, we define an integrated path loss distance.

\begin{definition} \label{def:pathLoss}
	\emph{\textbf{Integrated path loss distance:} We define the integrated path loss distance of the IRS-assisted cascaded channel (BS$\rightarrow$IRS$\rightarrow$UE) as $d_{\rm BIU}=d_{\rm BI} \circ d_{\rm IU}$, where $\circ$ depends on the adopted path loss law, $d_{\rm BI}$ and $d_{\rm IU}$ are the link distances of BS$\rightarrow$IRS and IRS$\rightarrow$UE links, respectively. If the product-distance path loss law is assumed, for near-field beamforming or far-field communications, the integrated path loss distance is defined as $d_{\rm BIU,P}=d_{\rm BI} \times d_{\rm IU}$; If the sum-distance path loss law is adopted, for near-field broadcasting, the integrated path loss distance is defined as $d_{\rm BIU,S}=d_{\rm BI} + d_{\rm IU}$ \cite{wu2021intelligent,tang2020wireless}.}
\end{definition}

The small-scale fading of the cascaded BS$\rightarrow$IRS$\rightarrow$UE channel over passive and active sub-surfaces, denoted by ${\mathbf{g}}_{\rm BIU}^{\rm pas}$, and ${\mathbf{g}}_{\rm BIU}^{\rm act}$, respectively, are modeled as
\begin{equation}\label{eq:pasacth}
	{\mathbf{g}}_{\rm BIU}^{\rm pas}= (\mathbf{g}_{\rm IU}^{\rm pas})^{H}{\bf \Psi}_{\rm pas}{\bf g}_{\rm BI}^{\rm pas}, \quad {\mathbf{g}}_{\rm BIU}^{\rm act}= (\mathbf{g}_{\rm IU}^{\rm act} )^{H}{\bf \Psi}_{\rm act}{\bf g}_{\rm BI}^{\rm act},
\end{equation}
where ${\bf g}_{\rm BI}^{\rm pas}\in\mathbb{C}^{N_{\rm pas}\times 1}$, $ (\mathbf{g}_{\rm IU}^{\rm pas})^{H}\in\mathbb{C}^{1\times N_{\rm pas}}$, ${\bf g}_{\rm BI}^{\rm act}\in\mathbb{C}^{N_{\rm act}\times 1}$, $ (\mathbf{g}_{\rm IU}^{\rm act})^{H}\in\mathbb{C}^{1\times N_{\rm act}}$, and we adopt parameter set $(l, j)$ to denote the link type and the mode of IRS sub-surface; $l \in \{{\rm IU}, {\rm BI}\}$ with $l =` {\rm IU}$' representing the IRS$\rightarrow$UE link and $l = `{\rm BI}$' indicating the BS$\rightarrow$IRS link. Similarly, $j \in \{{\rm pas}, {\rm act}\}$ with $j = `{\rm pas}$' denoting the passive IRS sub-surface and $j = `{\rm act}$' indicating the active IRS sub-surface. In (\ref{eq:pasacth}), $\mathbf{g}_{l}^{j}$ represents small-scale fading across the link $l$ of the employed IRS mode $j$, where $n_{j}\in\mathcal{N}_{j}\triangleq\{1,\cdots,N_{j}\}$, $l \in \{{\rm IU}, {\rm BI}\}$, and $j \in \{{\rm pas}, {\rm act}\}$. 

Furthermore, we consider generalized small-scale fading for each individual link to enable the versatility and flexibility of the results. It was displayed in \cite{6059452} that most of the existing fading distributions, i.e., Rayleigh, Rician, one-sided Gaussian, Nakagami-$m$, $\kappa-\mu$, and $\kappa-\mu$ shadowed fading, can be modeled as mixture Gamma distributions with high accuracy and tractability. To this end, we first introduce the definition of mixture Gamma distribution, and the properties of mixture Gamma distributed channels are summarized below \cite{devore1993constructive}. 
\begin{definition}\label{def:MG}
	\emph{\textbf{The mixture Gamma distribution:} A random variable $G$ that follows a mixture Gamma distribution is denoted as $G \sim \mathcal{MG}(\varepsilon_{i},\beta_{i},\xi_{i})$, and its PDF, CDF and moments are given by
		\begin{subequations} \label{eq:MG}
			\begin{alignat}{2}
				f_{\rm G}(x) =  & \sum_{i=1}^{I}\varepsilon_{i}x^{\beta_{i}-1}e^{-\xi_{i}x}, \label{eq: MG_PDF} \\
				F_{\rm G}(x) = & \sum_{i=1}^{I} \varepsilon_{i}\xi_{i}^{-\beta_{i}}\gamma(\beta_{i}, \xi_{i}x), \label{eq: MG_CDF}  \\
				\mathbb{E}_{\rm G}\left[ x^r  \right] = &  \sum_{i=1}^{I}  \varepsilon_{i} \frac{\Gamma(\beta_i+r)}{\xi_i^{\beta_i+r}},\label{eq: MG_moments} 
				% \quad \mathcal{L}_{\rm H}(s) = \sum_{i=1}^{I}  \varepsilon_{i} \frac{\Gamma(\beta_i)}{(\xi_i+s)^{\beta_i}},
			\end{alignat}
		\end{subequations}
		where $\{\varepsilon_{i},\beta_{i},\xi_{i}\}$ are the parameters for each Gamma component, $\gamma(\cdot,\cdot)$ is the lower incomplete Gamma function.\footnote{The number of Gamma components, $I$, is usually set as 20 to guarantee sufficient accuracy, i.e., approximation error less than $10^{-6}$ \cite{6059452,10458985}. }}
\end{definition}

Moreover, \cite{10458985} has proved that the cascaded small-scale channel gain can be modeled as a mixture Gamma distribution with high accuracy and tractability.

\begin{figure}[t!]
	\centering
	\includegraphics[width=\linewidth]{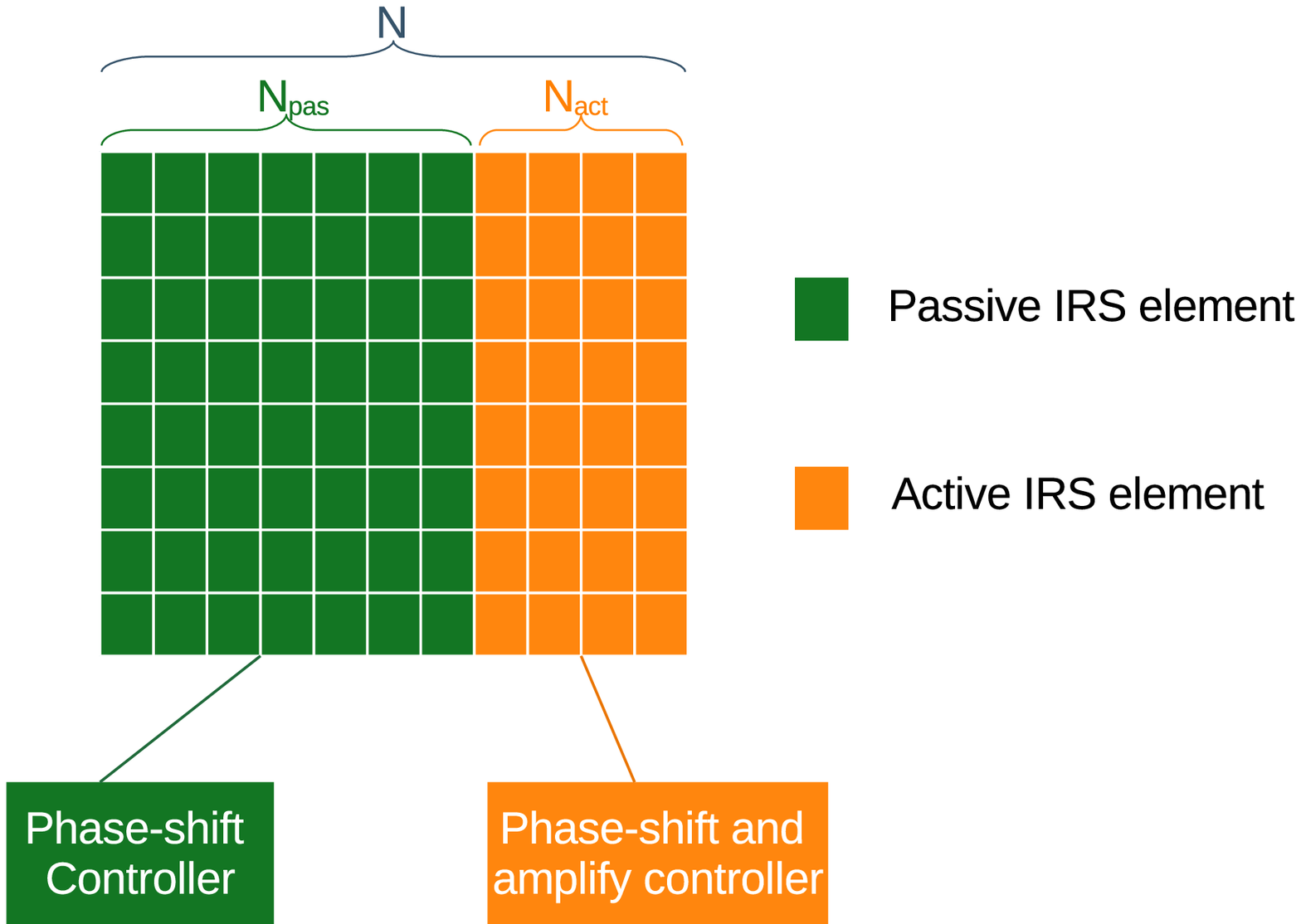}
	\caption{Hybrid IRS: One sub-surface contains $N_{\rm pas}$ passive reflecting elements along with a phase-shift controller, while the other contains $N_{\rm act}$ active reflecting elements, a phase-shift controller, and an amplify controller \cite{kang2023active}.}
	\label{fig:hybridIrs}
\end{figure}

\section{Geometric Models for the Integrated Path Loss Distance }  \label{sec:geo}

In this section, we study geometric models for the integrated path loss distance under product- and sum-distance path loss laws, as defined in Definition \ref{def:pathLoss}. We derive the CDF and PDF of the integrated path loss distance of the cascaded channel based on the proposed geometric models. Herein, we consider an end-to-end communication scenario, where the typical UE and its serving BS are symmetrically located along the x-axis at Cartesian coordinates $(c,0)$ and $(-c,0)$, respectively. We describe the geometric models for two path loss laws below.

\begin{figure*}[t]
	\centering
	\subfigure[Product-distance path loss law.]{
		\begin{minipage}[t]{0.47\textwidth}
			\includegraphics[width=\linewidth]{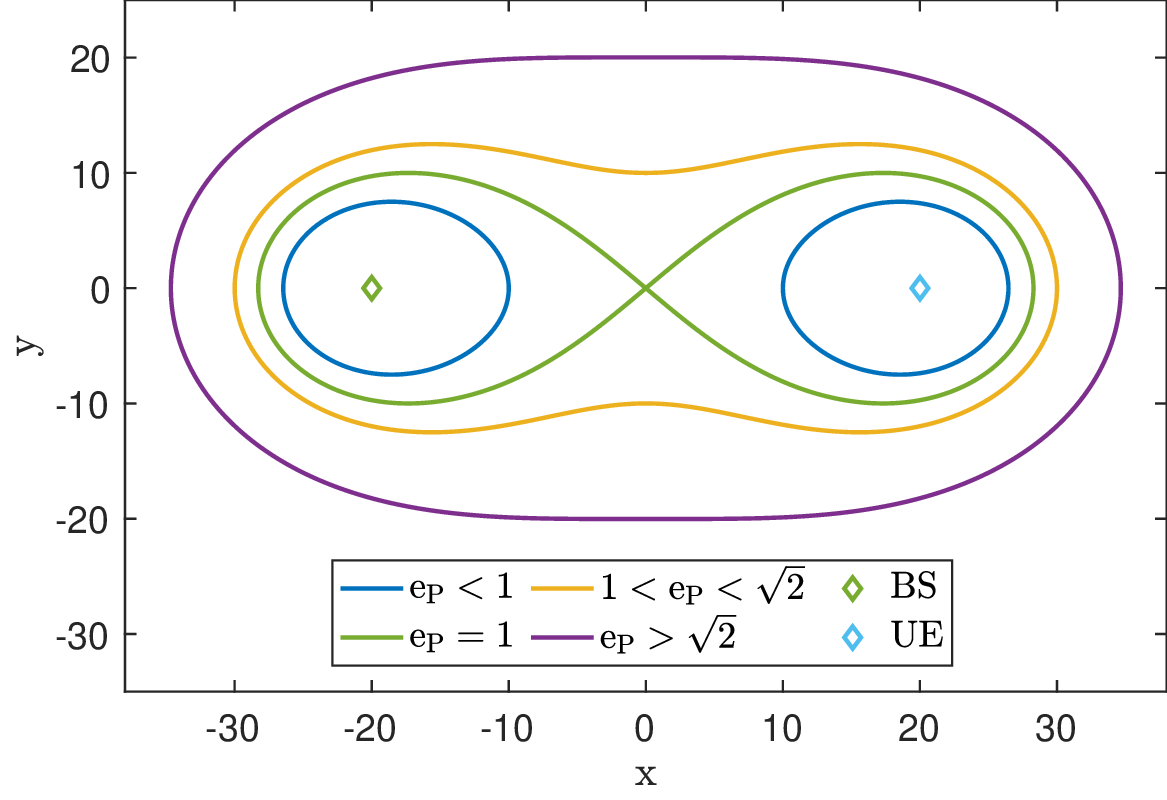}
		\end{minipage}
		\label{fig:cassiniOval}
	}
	\hspace{2pt}
	\subfigure[Sum-distance path loss law.]{
		\begin{minipage}[t]{0.47\textwidth}
			\includegraphics[width=\linewidth]{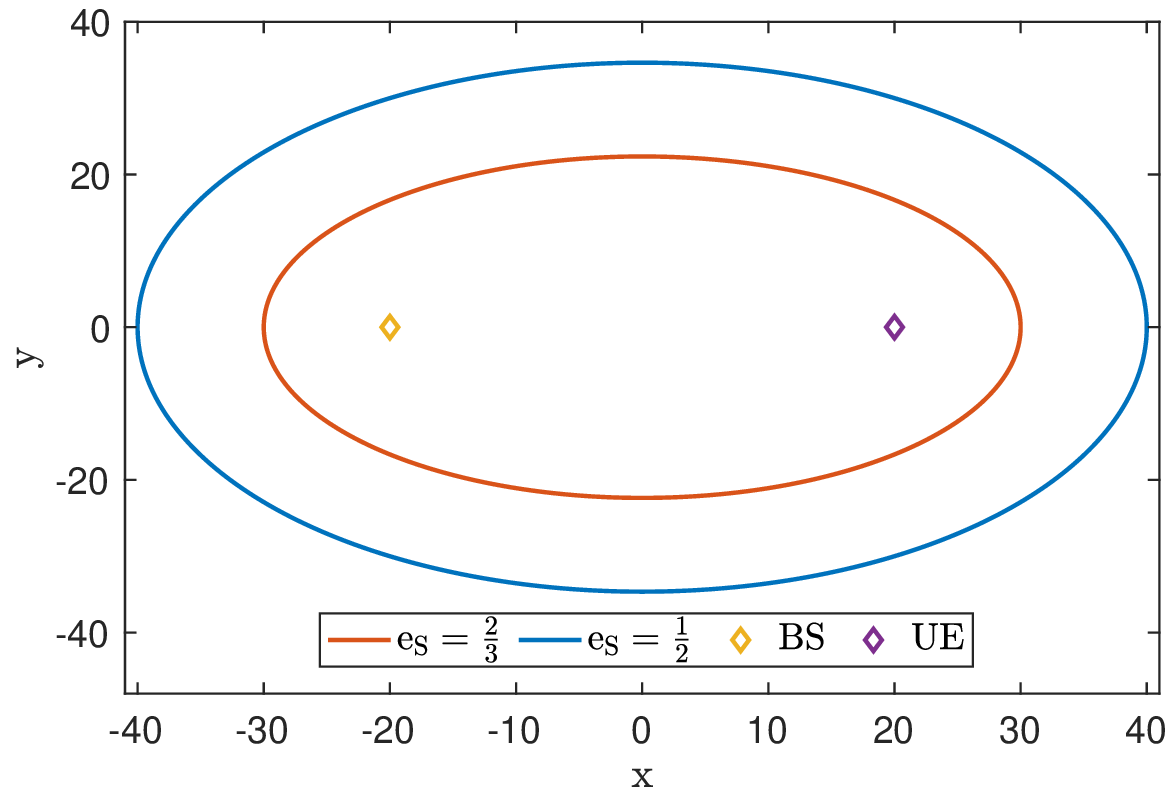}
		\end{minipage}
		\label{fig:ellipse}
	}
	\caption{Equal integrated path loss distance ($d_{\rm BIU}$) trajectories of IRS. }
	\label{fig:IRStrjectory}
\end{figure*}

\subsection{Product-Distance Path Loss Law: Cassini Oval}

For product-distance path loss law, the integrated path loss distance of the cascaded channel is defined as the product of individual link distance, i.e., $d_{\rm BIU,P} = d_{\rm BI}\times d_{\rm IU}$, which has been validated in \cite{tang2020wireless} to be suitable for both near-field beamforming and far-field communication scenarios.

\begin{theorem} \label{theorem:cassiniOval}
	\emph{\textbf{Equal product-distance trajectory (Cassini oval):}} For the product-distance path loss law,  the trajectory of IRSs with the equal integrated path loss distance of the cascaded BS$\rightarrow$IRS$\rightarrow$UE channel, $d_{\rm BIU,P}$, can be represented by a Cassini oval with the BS and UE as its foci, and its equation is given by
	\begin{equation} \label{eq:cassiniOval}
		(x^2+y^2)^2-2c^2(x^2-y^2)=d_{\rm BIU,P}^2-c^4,
	\end{equation}
	where the Cartesian coordinates of the typical UE and its serving BS are $ (\pm c,0) $, and $(x,y)$ is the Cartesian coordinates of the IRSs. The enclosed area of the equal product-distance trajectory is determined by \eqref{eq:areaOfCassiniOval}, denoted by $S(d_{\rm BIU,P})$,
	\begin{figure*}[b]
        \hrule
	\begin{equation}\label{eq:areaOfCassiniOval}
		\begin{split}
			S(d_{\rm BIU,P}) = 
			\begin{cases}
				2{c^2}\left[E\left(\frac{d_{\rm BIU,P}}{c^2}\right)-\left(1-\frac{d_{\rm BIU,P}^2}{c^4}\right)K\left(\frac{d_{\rm BIU,P}}{c^2}\right)\right],  & \text{for } {\rm e}_{\rm P} < 1 \\ 
				{2d_{\rm BIU,P}}E\left(\frac{c^2}{d_{\rm BIU,P}}\right), & \text{for } {\rm e}_{\rm P}\geq 1 
			\end{cases},
		\end{split}
	\end{equation}
\end{figure*}
where ${\rm e}_{\rm P}=\frac{\sqrt{d_{\rm BIU,P}}}{c}$, $K(\cdot)$ and $E(\cdot)$ are the complete elliptic integral of the first and second kind, respectively \cite{abramowitz1964handbook}.
\end{theorem}
\begin{proof}
See Appendix~\ref{append:theoremCassiniOval}.
\end{proof}

\begin{remark}\emph{
		The equal product-distance trajectory of the integrated path loss distance in Theorem~\ref{theorem:cassiniOval} exhibits varying shapes depending on the parameter, ${\rm e}_{\rm P}$, as shown in Fig. \ref{fig:cassiniOval}. In the following, we will discuss these four distinct shapes:
		\begin{itemize}
			\item If ${\rm e}_{\rm P}<1$, i.e., $d_{\rm BIU,P} < c^2$, the trajectory consists of two foci-centered disconnected loops, the IRS with smallest $d_{\rm BIU,P}$ is supposed to be located in the vicinity of the BS or typical UE with a preference for being positioned between transceivers. This observation differs from the results obtained from the optimization perspective in \cite{wu2021intelligent}. In \cite{wu2021intelligent}, the IRSs are constrained on the segment between the typical UE and its serving BS, within which the area outside the segment is not investigated, and it is a one-dimensional special case of this Cassini oval model.
			\item If ${\rm e}_{\rm P}=1$, i.e., $d_{\rm BIU,P} = c^2$, the trajectory is a lemniscate of Bernoulli with the shape of a sideways figure eight and a double point at the middle-point of the BS and typical UE. The presence of double point in the trajectory causes the breakdown of area function in Theorem~\ref{theorem:cassiniOval} into piecewise functions, producing an impulse on the PDF of $d_{\rm BIU}$, as verified numerically in Fig.~\ref{fig:PDF_cassiniOval}.
			\item If ${\rm e}_{\rm P}>1$, i.e., $d_{\rm BIU,P} > c^2$, the trajectory is a continuous loop. When $1<{\rm e}_{\rm P}<\sqrt{2}$, the trajectory shapes like a peanut; When ${\rm e}_{\rm P} \geq\sqrt{2}$, the trajectory is a convex loop, and the vicinity preference of IRS vanishes, which is similar to the case of sum-distance path loss law. 
	\end{itemize}
    Therefore, the individual link distance of the IRS$\rightarrow$UE link $d_{\rm IU}$ or BS$\rightarrow$IRS link $d_{\rm BI}$ does not solely determine the integrated path loss distance $d_{\rm BIU}$. As shown in Fig.~\ref{fig:IRStrjectory}, for both product- and sum-distance path loss laws, the $d_{\rm BIU}$ varies with the locations of the IRS for the same $d_{\rm IU}$ or $d_{\rm BI}$. These geometric models offer robust theoretical support for the IRS deployment findings in \cite{moustakas2023reconfigurable,ren2023deployment}, and elucidates the IRS's distinct behavior at varying elevations when the IRS is located between the BS and UE, which are also one-dimensional instances of the proposed Cassini oval model.
}
\end{remark}

Using Theorem~\ref{theorem:cassiniOval}, we derived the PDF and CDF of the integrated path loss distance $d_{\rm BIU}$ for the product-distance path loss model, as shown below.

\begin{lemma} \label{lem:cassiniOvalPCDF}
	The CDF and PDF of the product distance $d_{\rm BIU,P}$ are respectively given by \eqref{eq:cassiniCDF} and \eqref{eq:cassiniPDF}, when IRS is located within the Cassini oval with maximum integrated path loss distance $d_{\rm t}$,
	\begin{figure*}
	\begin{equation} \label{eq:cassiniCDF}
		\begin{split}
			F_{\rm P}(d_{\rm BIU,P}) =  
			\begin{cases}
				\frac{2c^2}{S(d_{\rm t})}\left[E\left(\frac{d_{\rm BIU,P}}{c^2}\right)-\left(1-\frac{d_{\rm BIU,P}^2}{c^4}\right)K\left(\frac{d_{\rm BIU,P}}{c^2}\right)\right],  & \text{for } {\rm e}_{\rm P} < 1 \\ 
				\frac{2d_{\rm BIU,P}}{S(d_{\rm t})}E\left(\frac{c^2}{d_{\rm BIU,P}}\right), & \text{for } {\rm e}_{\rm P}\geq 1 
			\end{cases}, 
		\end{split}
	\end{equation}
\hrule
\end{figure*} 
	\begin{equation}\label{eq:cassiniPDF}
		\begin{split}
			f_{\rm P}(d_{\rm BIU,P}) = 
			\begin{cases}
				\frac{{2}d_{\rm BIU,P}}{c^2{S(d_{\rm t})}}K\left(\frac{d_{\rm BIU,P}}{c^2}\right),  & \text{for } {\rm e}_{\rm P} < 1, \\ 
				\frac{2}{S(d_{\rm t})}K\left(\frac{c^2}{d_{\rm BIU,P}}\right), & \text{for } {\rm e}_{\rm P} \geq 1. 
			\end{cases} 
		\end{split}
	\end{equation}
The enclosed area of $d_{\rm t}$, denoted by $S(d_{\rm t})$, is defined in \eqref{eq:areaOfCassiniOval}, which is based on the Cassini oval model in Theorem \ref{theorem:cassiniOval}.
\end{lemma}
\begin{proof}
	See Appendix~\ref{append:lemma1}.
\end{proof}

\subsection{Sum-Distance Path Loss Law: Ellipse} 

The integrated path loss distance of the cascaded channel for the sum-distance path loss model is defined as the sum of individual link distance, i.e., $d_{\rm BIU,S} = d_{\rm BI}+d_{\rm IU}$, validated in \cite{tang2020wireless} to be suitable for the near-field broadcasting scenario.

\begin{proposition} \label{pro:ellipse}
	\emph{\textbf{Equal sum-distance trajectory (Ellipse):}} For the sum-distance path loss law, the trajectory of IRSs with the equal integrated path loss distance of the cascaded BS$\rightarrow$IRS$\rightarrow$UE channel, $d_{\rm BIU,S}$, follows an ellipse with the BS and UE as its foci, as expressed below
 % \begin{small}
	\begin{equation} \label{eq:ellipse}
		\left(\Delta-c^2\right) x^2+\Delta y^2=\Delta\left(\Delta-c^2\right),
        \quad \Delta \triangleq \frac{d_{\rm BIU,S}^2}{4},
	\end{equation}
	% \end{small}
 where the Cartesian coordinates of the typical UE and its serving BS are $ (\pm c,0) $, and $(x,y)$ is the Cartesian coordinates of the IRSs. The enclosed area of the equal sum-distance trajectory is determined by  \cite{berger2009geometry}, denoted as $S(d_{\rm BIU,S})$,
	\begin{equation}\label{eq:areaOfEllipse}
		S(d_{\rm BIU,S}) = \frac{\pi}{4} d_{\rm BIU,S}\sqrt{d_{\rm BIU,S}^{2}-4c^{2}}.
	\end{equation}
\end{proposition}

Fig. \ref{fig:ellipse} shows the equal sum-distance trajectory for sum-distance path loss model, where ${\rm e}_{\rm S} = \frac{2c}{d_{\rm BIU,S}}$. The equal trajectories of the sum-distance path loss law are confocal ellipses. In the following Lemma, we derived the CDF and PDF of the integrated path loss distance for the sum-distance path loss law using the equations in Proposition~\ref{pro:ellipse}.

\begin{lemma} \label{lem:ellipsePCDF}
	Based on Proposition \ref{pro:ellipse}, the CDF and PDF of the product distance $d_{\rm BIU,S}$ can be expressed as \eqref{eq:ellipseCDF} and \eqref{eq:ellipsePDF}, respectively, when IRS is located within the Ellipse with maximum integrated path loss distance $d_{\rm t}$.
	%\begin{small}
	\begin{equation} \label{eq:ellipseCDF}
		F_{\rm S}(d_{\rm BIU,S}) =  \frac{\pi d_{\rm BIU,S}\sqrt{d_{\rm BIU,S}^{2}-4c^{2}}}{4S(d_{\rm t})},
	\end{equation}
%\end{small}
	%\begin{small}
		\begin{equation} \label{eq:ellipsePDF}
		f_{\rm S}(d_{\rm BIU,S}) =  \frac{\pi}{2S(d_{\rm t})}\frac{d_{\rm BIU,S}^{2}-2c^2}{\sqrt{d_{\rm BIU,S}^{2}-4c^2}}.
	\end{equation}
%\end{small}
The enclosed area of $d_{\rm t}$, denoted by $S(d_{\rm t})$, is defined in \eqref{eq:areaOfEllipse}.
\end{lemma}

\subsection{Discussions on the Geometric Models}

Building upon the geometric models, several observations have been made regarding the integrated path loss distance of the cascaded channel. First, the IRS locations with the smallest $d_{\rm BIU}$ are dependent on both $d_{\rm IU}$ and the relative positions of IRS and UE/BS, as shown in Fig. \ref{fig:IRStrjectory}. Second, as demonstrated by the geometric models, there are instances where the preference for passive IRS location in the vicinity of transceivers is absent. Specifically, under product-distance path loss law with ${\rm e}_{\rm P}\geq 1$ and sum-distance path loss law, the preference of the vicinity of transceivers does not persist, which is different from the product-distance path loss law with ${\rm e}_{\rm P}< 1$. Even though the product-distance path loss law maintains a preference for proximity if ${\rm e}_{\rm P}<1$, this preference is not uniform and diminishes with increasing ${\rm e}_{\rm P}$.

The geometric models offer valuable insights into the deployment strategies of passive IRS. First, the conventional notion of deploying IRSs in the vicinity of transceivers is challenged, especially when the direct communication link is blocked. The reasons are as follows: a) The trajectories of equal integrated path loss distances do not conform to uniform circular patterns with communication nodes as their centers. b) Environmental factors contribute to coherence issues, leading to compromised link performance when the IRS is located near transceivers due to the substantial blockage of the direct path. Second, the deployment strategies for IRS can be judiciously designed to achieve optimal utility while concurrently upholding a desirable outage probability. This is particularly relevant in scenarios involving fixed transmitter and receiver positions, such as M2M communication. Notably, the trajectories of equal integrated path loss distances are contingent upon their specific spatial placements, lending an inherent advantage to strategic deployment planning. 

In multi-UE scenarios, such as \cite{zhang2021intelligent,he2023joint}, the IRS deployment strategy can be meticulously designed based on the geometric models. To illustrate the applications of the geometric models in deployment strategies, we introduce the concept of IRS effective zone for a transceiver pair characterized by the IRS location-based long-term performance. 
	\begin{definition} \label{def:IrsSignalZone}
		\emph{\textbf{IRS Effective Zone for a transceiver pair:} To ensure reliable communication between the typical UE and its serving BS, the maximum integrated path loss distance of the cascaded channel via passive IRS is defined as $d_{\rm BIU,max}=\left( \frac{P_{\rm T}N^2\mathbb{E}[G_{\rm BIU}]}{\tau\delta^2} \right)^{\frac{1}{\alpha}}$ to meet the outage signal-to-noise ratio (SNR) threshold, $\tau$, where $\mathbb{E}[G_{\rm BIU}]$ is the average small-scale channel gain of the BS$\rightarrow$IRS$\rightarrow$UE channel. The area enclosed by the trajectory corresponding to $d_{\rm BIU,max}$ is defined as the IRS effective zone for a transceiver pair.}
	\end{definition}
As depicted in Fig.~\ref{fig:multipleCellNetwork}, the IRSs are divided into three types based on their locations. Therefore, in multi-UE scenarios, the association of IRS can be customized to various objective functions. For instance, considering the utility of IRS, we can opt for the Type-II IRS, capable of providing satisfactory performance for both UEs. Alternatively, concerning interference mitigation among UEs, the Type-I IRS may be selected, exclusively serving one UE.

\section{Optimal Placement and Opportunistic Association Policy for Hybrid IRS} \label{sec:assoc}

In this section, we first display the channel power statistics of the integrated IRS-aided cascaded channel, then determine the optimal placement of hybrid IRS, and last present an opportunistic association policy in terms of maximizing the long-term performance of the IRS-aided communication \cite{xia2012cooperative}.

\subsection{Channel Power Statistics}

We assume that the phase shift of all reflecting elements are aligned to enhance the transmitted signals \cite{wu2019intelligent}. Thus, the received signal is given by
\begin{equation}
\begin{split}
    	y &= \left( \mathbf{g}_{\rm BIU}^{\rm pas} + \mathbf{g}_{\rm BIU}^{\rm act}\right) \epsilon d_{\rm BIU}^{-{\alpha}/{2}}  x \\
     &+(\mathbf{g}_{\rm IU}^{\rm act} )^{H}{\bf \Psi}_{\rm act}\sqrt{\epsilon} d_{\rm IU}^{-{\alpha}/{2}}\mathbf{n_{\rm F}}+n_{0},
\end{split}
\end{equation} 
and the received SNR is given by
\begin{equation}
	\mathbb{E}[ {\rm SNR} ] = \mathbb{E} \left[ \frac{P_{\rm T}G_{\rm BIU}d_{\rm BIU}^{-\alpha}(N_{\rm pas}+\eta N_{\rm act})^{2}}{N_{\rm act}\eta^{2}\delta_{\rm F}^{2}G_{\rm IU} d_{\rm IU}^{-\alpha}+\delta^{2} }\right],      
\end{equation}
where $\mathbf{n_{\rm F}}\in\mathbb{C}^{N_{\rm act}\times 1}$ is the thermal noise generated by the active part, which is assumed to follow the Complex Normal distribution, i.e., $\mathbf{n_{\rm F}}\sim\mathcal{CN}(\mathbf{0}_{N_{\rm act}},\delta_{\rm F}^{2}\mathbf{I}_{N_{\rm act}})$, ${n_{0}}$ is the additive white Gaussian noise (AWGN) at the user, i.e., ${n_{0}}\sim\mathcal{N}(0,\delta_{0}^{2})$, and $G_{\rm BIU}$ ($G_{\rm IU}$) is the small-scale channel gain of the BS$\rightarrow$IRS$\rightarrow$UE (IRS$\rightarrow$UE) channel. 

Herein, for expression concise, we assume all reflecting elements suffer the same small-scale fading, and all transmission links follow Nakagami-$m$ fading, i.e., $ |g_{\rm BI}|$ and $ |g_{\rm IU}|$ follow Nakagami-$m$ distributions with parameters $m_{\rm BI}$ and $m_{\rm IU}$, respectively. The analysis can be extended to generalized fading straightforwardly as the small-scale fading is assumed independent of the link distance.  

Given Nakagami-$m$ distributed fading, the channel gain of BS$\rightarrow$IRS and IRS$\rightarrow$UE links can be represented as Gamma distributions with unit mean. The distribution of small-scale channel gain for individual links, i.e., BS$\rightarrow$IRS and IRS$\rightarrow$UE, are given by $G_{\rm BI} \sim \mathcal{MG}\left(\varepsilon_{{\rm BI},i},\beta_{{\rm BI},i},\xi_{{\rm BI},i}\right)$ and $G_{\rm IU} \sim \mathcal{MG}\left(\varepsilon_{{\rm IU},i},\beta_{{\rm IU},i},\xi_{{\rm IU},i}\right)$, where
\begin{equation} \label{eq:MG_para_individual}
	\begin{split}
		& \varepsilon_{{\rm BI},i} = \frac{m_{\rm BI}^{m_{\rm BI}}}{\Gamma(m_{\rm BI})},\quad \beta_{{\rm BI},i}=\xi_{{\rm BI},i}=m_{\rm BI}, \\
		& \varepsilon_{{\rm IU},i} = \frac{m_{\rm IU}^{m_{\rm IU}}}{\Gamma(m_{\rm IU})},\quad \beta_{{\rm IU},i}=\xi_{{\rm IU},i}=m_{\rm IU}.
	\end{split}
\end{equation}

Moreover, the multiplicability of the mixture Gamma distribution, as proven in \cite{10458985}, allows for modeling the integrated channel gain of the cascaded BS$\rightarrow$IRS$\rightarrow$UE channel as a mixture Gamma distribution, denoted as $G_{\rm BIU}\triangleq G_{\rm BI} G_{\rm IU}  \sim \mathcal{MG}(\varepsilon_{{\rm BIU},i},\beta_{{\rm BIU},i},\xi_{{\rm BIU},i})$, and the parameters are given by
\begin{align}\label{eq:MG_para_BIU}
	\begin{split}
		&~ \varepsilon_{{\rm BIU},i} = \frac{(m_{\rm BI}m_{\rm IU})^{m_{\rm BI}}w_{i}t_{i}^{m_{\rm IU}-m_{\rm BI}-1}}{\Gamma(m_{\rm BI})\Gamma(m_{\rm IU})}, \\ &~
		\beta_{{\rm BIU},i} = m_{\rm BI}, \quad
		\xi_{{\rm BIU},i} = \frac{m_{\rm BI}m_{\rm IU}}{t_{i}},
	\end{split}
\end{align} 
where $t_{i}$ is the $i$-th root of the Laguerre polynomial $L_p(t)$, and $w_{i}$ is the $i$-th weight of the Gaussian-Laguerre quadrature $\int_0^\infty e^{-t}f(t)dt \approx \sum_{i=1}^{I} w_{i} f\left(t_i\right)$ that is defined as $w_{i} = \frac{t_i}{\left(p+1\right)^2 L_{p+1}\left(t_i\right)^2}$ \cite{abramowitz1964handbook}.

\subsection{Optimal Locations of the Hybrid IRS}

To determine the optimal locations of the hybrid IRS, taking into account the long-term average received SNR, the problem is formulated as below 
%\begin{small}
\begin{subequations}
\begin{alignat}{2}
	({\rm P}1):~\max_{\substack {d_{\mathrm {BI}}, d_{\mathrm {IU}}} } \,&  \mathbb{E}[ {\rm SNR} ]  \notag\\ 
	\mathrm {s.t.}~~ & P_{\rm T}\geq~0,~  P_{\rm F} ~\geq 0, \label{con:P} \\ 
	& \eta^{2} \leq \frac{P_{\rm F}}{N_{\rm act}(P_{\rm T}G_{\rm BI}d_{\rm BI}^{-\alpha} +\delta_{\rm F}^{2})}, \label{con:eta} \\  
	& P_{\rm F} \geq N_{\rm act}(P_{\rm T}G_{\rm BI}d_{\rm BI}^{-\alpha} +\delta_{\rm F}^{2}) ,\label{con:eta1}  \\ 
	&d_{\rm BIU}  =  d_{\rm BI}\circ d_{\rm IU} ,\label{con:pathL}  \\  
	& d_{\rm BI} \geq 0,\quad d_{\rm IU}  \geq~  0,\label{con:dist} 
 \end{alignat}
\end{subequations}
%\end{small}
where \eqref{con:eta} and \eqref{con:eta1} represent the hardware limit of the active reflecting elements and the power constraint for the amplification factor \cite{amato2018tunneling}. Note that, the inequality in \eqref{con:eta} holds for general cases, including dynamic $d_{\rm BI}$, and the equal symbol holds for associated IRS. Furthermore, \eqref{con:eta1} indicates  $\eta\geq 1$ \cite{kang2023active}. 

The optimization in (P1) is a challenging problem due to the coupled relation between the thermal noise and the amplitude amplification factor. The thermal noise is generated by active parts and the amplification factor depends on the channel condition of the BS$\rightarrow$IRS link, which affects both the received signal power and thermal noise at the UE side. To simplify this problem, we note that (P1) includes the conventional IRS architectures as two special cases. Specifically, if $N_{\rm act}=0$, the hybrid IRS reduces to the conventional passive IRS, and the solution of (P1) is the optimal locations of passive IRS. Otherwise, the object of (P1) is to find the optimal locations of the active IRS when $N_{\rm pas}=0$.

Thereby, we opt to initially partition this problem into two distinct parts, aiming to ascertain the optimal locations for the passive and active sub-surfaces individually. The optimal locations for the passive sub-surface are determined as points along the trajectory possessing the smallest value of $d_{\rm BIU}$, where each $d_{\rm BIU}$ represents a trajectory consisting of potential IRS locations that yield identical performance. In contrast, performance of the active sub-surface varies along this trajectory due to the amplification factor and the additional thermal noise generated by the active parts. 

Hence, our initial focus lies on determining the optimal locations of the active sub-surface on a given trajectory defined by $d_{\rm BIU}$, and the problem is formulated as follows
\begin{alignat}{2}
 	({\rm P}2):~ \max_{\substack { d_{\mathrm {BI}}, d_{\mathrm {IU}}} }& ~ \mathbb{E}[ {\rm SNR} ]^{(\rm act)} \quad 
	 \mathrm {s.t.} \quad \eqref{con:P}-\eqref{con:dist},\notag 
\end{alignat}
where $ \mathbb{E}[ {\rm SNR} ]^{(\rm act)}$ is given by
\begin{equation}
	 \mathbb{E}[ {\rm SNR} ]^{(\rm act)} =  \mathbb{E} \left[ \frac{P_{\rm T}G_{\rm BIU}d_{\rm BIU}^{-\alpha}\eta^{2} N_{\rm act}^{2}}{N_{\rm act}\eta^{2}\delta_{\rm F}^{2}G_{\rm IU} d_{\rm IU}^{-\alpha}+\delta^{2} }\right]. 
\end{equation}
Then, with the maximum value of the amplification factor, $ \mathbb{E}[ {\rm SNR} ]^{(\rm act)}$ can be derived as \eqref{eq:meanSNR_act}, where the mean value of the moments of the small-scale fading can be derived by substituting \eqref{eq:MG_para_individual}, \eqref{eq:MG_para_BIU} and \eqref{eq: MG_moments}, as described below
\footnote{Note that in this study, we employ Nakagami-$m$ fading for expression conciseness. It is worth mentioning that the outcomes derived from our approach can be extended to generalized fading, which is versatile for application across various frequency bands \cite{6059452}.}
%\begin{small}
\begin{figure*}[b!]
	\hrule
\begin{equation} \label{eq:meanSNR_act}
	 \mathbb{E}[ {\rm SNR} ]^{(\rm act)} = \frac{P_{\rm T}P_{\rm F}N_{\rm act}d_{\rm BIU}^{-\alpha} }{P_{\rm F}\delta_{\rm F}^{2}d_{\rm IU}^{-\alpha}\mathbb{E}\left[ {G_{\rm BI}}^{-1} \right]+P_{\rm T}\delta^{2} d_{\rm BI}^{-\alpha}\mathbb{E}\left[{G_{\rm IU}}^{-1}\right]+\delta^{2}\delta_{\rm F}^{2} \mathbb{E}\left[{G_{\rm BIU}}^{-1} \right] },
\end{equation}
\end{figure*}
\begin{equation}
	\begin{split}
		&\mathbb{E}\left[{G_{\rm BI}}^{-1}\right] ~ =   \frac{m_{\rm BI}}{m_{\rm BI}-1} , \quad\mathbb{E}\left[ {G_{\rm IU}}^{-1}\right] ~ = \frac{m_{\rm IU}}{m_{\rm IU}-1}, \\
		&\mathbb{E}\left[ {G_{\rm BIU}}^{-1} \right] = \sum_{i=1}^{I} \frac{m_{\rm BI}m_{\rm IU} w_{i} t_{i}^{m_{\rm IU}-1} }{(m_{\rm BI}-1)\Gamma(m_{\rm IU})}.
	\end{split} 
\end{equation}
%\end{small}
Thus, (P2) can be further simplified as 
%\begin{small}
\begin{align*}
	({\rm P}3):~ \min_{\substack { d_{\mathrm {BI}}, d_{\mathrm {IU}}} }& ~ f = P_{\rm F}\delta_{\rm F}^{2}d_{\rm BI}^{\alpha}\mathbb{E}\left[{G_{\rm BI}}^{-1}\right]+P_{\rm T}\delta^{2} d_{\rm IU}^{\alpha} \mathbb{E}\left[{G_{\rm IU}}^{-1}\right]\\
	\mathrm {s.t.}~&~% d_{\rm BI}\geq~ \left( \frac{N_{\rm act}P_{\rm T}}{P_{\rm F}-N_{\rm act}\delta_{\rm F}^{2}}\right) ^{\frac{1}{\alpha}} \\  &	
	\eqref{con:P},  
	\eqref{con:eta1}-\eqref{con:dist}.\notag 
\end{align*}
%\end{small}
Next, by employing the two path loss laws, the results are achieved as below,\footnote{For the product-distance path loss law, when the power and noise level of both the BS and the active IRS are comparable, the optimal solution lies at the line that is perpendicular to and divides the segment made by BS and UE into two equal halves. This observation is consistent with relay systems. Besides, the feasibility conditions can offer theoretical backing for system design as well.} 
%taking the derivative of $f$ with respect to $d_{\rm BI}^{\alpha}$, the problem is solved, and 
%\begin{small}
\begin{alignat}{2} \label{eq:d_BI_opt}
	d_{\rm BI}^{*} = ~& \rho\sqrt{d_{\rm BIU}}, && \quad \text{for product},  \notag \\
	d_{\rm BI}^{*} = ~& \frac{d_{\rm BIU}}{\varphi_{\rm S}+1},   && \quad \text{for sum}	, 
\end{alignat}
%\end{small}
where $\rho=\left({\frac{P_{\rm T}\delta^{2} \mathbb{E}\left[{G_{\rm IU}}^{-1}\right] }{P_{\rm F}\delta_{\rm F}^{2}\mathbb{E}\left[{G_{\rm BI}}^{-1}\right]}  } \right)^{\frac{1}{2\alpha}}$, and $\varphi_{\rm S}=\rho^{\frac{2\alpha}{1-\alpha}}$. %{\left(\frac{P_{\rm F}\delta_{\rm F}^{2}}{P_{\rm T}\delta^{2}}\right)^{\frac{1}{\alpha-1}}}$.
Thereby, the corresponding distance between IRS and UE can be readily given by
\begin{alignat}{2}  \label{eq:d_IU_opt}
	d_{\rm IU}^{*} = ~&  \rho^{-1}\sqrt{d_{\rm BIU}},  &&\quad \text{for product}, \notag\\
	d_{\rm IU}^{*} = ~& d_{\rm BIU}-(d_{\rm BI})^{*}, &&  \quad \text{for sum}	. 
\end{alignat}

It can be verified that (P3) is feasible if and only if the following constraints of system parameters are satisfied \cite{you2021wireless}.
% \begin{small}
\begin{alignat}{2}  \label{eq:etaGeq1}
	&\rho~\geq \left(\frac{N_{\rm act}P_{\rm T}}{P_{\rm F}-N_{\rm act}\delta_{\rm F}^{2}} \right) ^{\frac{1}{\alpha}}d_{\rm BIU,P}^{-\frac{1}{2}},  &&\quad \text{for product}, \notag\\
	&\varphi_{\rm S}\leq \left(\frac{N_{\rm act}P_{\rm T}}{P_{\rm F}-N_{\rm act}\delta_{\rm F}^{2}} \right) ^{-\frac{1}{\alpha}}d_{\rm BIU,S}-1, &&  \quad \text{for sum}	. 
\end{alignat}
%\end{small}
As illustrated in Fig.~\ref{fig:cassiniOval}, the feasibility of the local optimal points for product-distance path loss law in \eqref{eq:d_BI_opt} varies along ${\rm e}_{\rm P}$, due to the various trajectory shapes. Specifically, the optimal points for product-distance path loss law on each trajectory are available if and only if the system parameters are designed based on the following inequality
\begin{equation} \label{eq:feasibility}
    \rho+\frac{1}{\rho} \geq \frac{2}{{\rm e}_{\rm P}}.   
\end{equation}
If ${\rm e}_{\rm P}\geq1$, the inequality in \eqref{eq:feasibility} is always satisfied and the results in \eqref{eq:d_BI_opt} are robust. However, if ${\rm e}_{\rm P}<1$, the value of $d_{\rm BI}^{*}$ may be infeasible in certain parameter configurations, i.e., $\rho+{1}/{\rho} < 2/{\rm e}_{\rm P}$. Then, the minimum $d_{\rm BI}$ on each trajectory is sub-optimal, considering that the function $f$ is convex and active IRS is supposed to be deployed near UE instead of BS due to practical operation power of BS and IRS \cite{you2021wireless}. Intuitively, as ${\rm e}_{\rm P}$$\rightarrow$0, the IRS is deployed in the vicinity of transceivers and active IRS should operate at lower power mode. 
\begin{remark}
	\emph{The optimal locations of the active sub-surface along a predetermined trajectory with a fixed $d_{\rm BIU}$ rely on $\rho$ and $\alpha$. This finding is consistent with \cite{you2021wireless}, where the active IRS is confined to the segment between UE and BS with limited height. }
\end{remark}

In the following, we assume that the system is designed based on \eqref{eq:etaGeq1} and \eqref{eq:feasibility}. With the results in \eqref{eq:d_BI_opt}, \eqref{eq:d_IU_opt} and the geometric models studied in Section \ref{sec:geo}, the optimal locations of the active sub-surface are derived in the following lemma.

\begin{lemma}\label{lem:localOp}
	When the typical UE and its serving BS are located at $(\pm d_{\rm BU}/2,0)$, the coordinates of the optimal locations of the active sub-surface are given by
    %\begin{small}
	\begin{alignat}{2}\label{eq:cood}
		x_{1}^{*} = & \frac{\varphi_{\rm P}^{2}-d_{\rm BIU}^{2}}{2d_{\rm BU}\varphi_{\rm P}}, ~ &y_{1}^{*} = & \pm \sqrt{ \varphi_{\rm P} - (x^{*}+c)^2 }, \notag \\ x_{2}^{*} = & \frac{d_{\rm BU}^{2}-d_{\rm BIU}^{2}\left(1-\frac{2}{\varphi_{\rm S}}\right)}{2d_{\rm BU}}, ~
		&y_{2}^{*} = & \pm \sqrt{\left(\frac{d_{\rm BIU}}{\varphi_{\rm S}+1}\right)^2-(x^{*})^2}, 
\end{alignat}
%\end{small}
where $\varphi_{\rm P}=\rho^{2}d_{\rm BIU}$, $(x_{1}^{*},y_{1}^{*})$ and $(x_{2}^{*},y_{2}^{*})$ are the optimal coordinates of the active sub-surface for the product- and sum-distance path loss laws, respectively.
\end{lemma}
\begin{proof}
	See Appendix~\ref{append:localOp}.
\end{proof}

However, the results in Lemma~\ref{lem:localOp} are local solutions for a predetermined $d_{\rm BIU}$. Subsequently, we proceed to acquire the globally optimal positions of active sub-surface by searching all the local optimal solutions across varying $d_{\rm BIU}$, the problem of determining the globally optimal solutions for (P2) is formulated as\footnote{Because the small-scale fading has been incorporated in the optimal solutions in \eqref{eq:d_BI_opt} and \eqref{eq:d_IU_opt} as a scaling factor of the path loss, the following objective functions omit the small-scale fading parameter \cite{chun2017comprehensive}.} 
\begin{equation*}
	\begin{split}
		({\rm P}4):~ \max_{\substack {d_{\mathrm {BIU}}} }& ~ \frac{P_{\rm T} P_{\rm F}  N_{\rm act}^{2} d_{\mathrm {BIU}}^{-\alpha} }{P_{\rm F} \delta_{\rm F}^{2} (d_{\rm IU}^{*})^{-\alpha} + P_{\rm T} \delta^{2} (d_{\rm BI}^{*})^{-\alpha} +\delta^{2} \delta_{\rm F}^{2}  }  \\
		\mathrm {s.t.}&~   \eqref{con:P}, \eqref{con:pathL} -\eqref{con:dist},\eqref{eq:d_BI_opt},\eqref{eq:d_IU_opt}. \notag  
	\end{split}
\end{equation*}

Next, we applied the results in \eqref{eq:d_BI_opt} and \eqref{eq:d_IU_opt} to (P4), which means the optimal locations of IRS for each determined $d_{\rm BIU}$ are selected. When the product-distance path loss law is adopted, (P4) can be simplified to 
\begin{equation*}
	\begin{split}
		({\rm P}5):~ \max_{\substack {d_{\mathrm {BIU}}} }& ~  \frac{P_{\rm T} P_{\rm F}  N_{\rm act} }{2\sqrt{ P_{\rm F}P_{\rm T}  \delta_{\rm F}^{2} \delta^{2}  } d_{\mathrm {BIU}}^{\frac{\alpha}{2}}   + \delta^{2} \delta_{\rm F}^{2} d_{\mathrm {BIU}}^{\alpha} } .
	\end{split}
\end{equation*}
When the sum-distance path loss law is adopted, (P4) can be simplified to 
\begin{equation*}
	({\rm P}6):~ \max_{\substack {d_{\mathrm {BIU}}} } ~\frac{P_{\rm T} P_{\rm F}  N_{\rm act}}{ P_{\rm F}\delta_{\rm F}^{2}\left( \frac{\varphi_{\rm S}}{\varphi_{\rm S}+1} \right) ^{-\alpha} +  P_{\rm T} \delta^{2} \left( \frac{1}{\varphi_{\rm S}+1} \right) ^{-\alpha}  +  \delta_{\rm F}^{2}d_{\rm BIU}^{\alpha} }.
\end{equation*}

The singularity of both (P5) and (P6) with respect to $d_{\rm BIU}$ clearly indicates that the optimal solutions exist at the minimum value of $d_{\rm BIU}$ while still satisfying the condition for optimal placement of the passive sub-surface. Therefore, the optimal locations of the hybrid IRS are given by 
\begin{alignat}{2}  \label{eq:hybrid_IRS_opt}
	d_{\rm BI}^{**} = & \rho\sqrt{d_{\rm BIU,min}}, ~ d_{\rm IU}^{**} =  \rho^{-1}\sqrt{d_{\rm BIU,min}},   &&   \text{for product},  \notag \\
	d_{\rm BI}^{**} = & \frac{d_{\rm BIU,min}}{\varphi_{\rm S}+1}, ~ \quad ~ d_{\rm IU}^{**} =  d_{\rm BIU,min}-{d_{\rm BI}}^{**}, &&  ~ \text{for sum}.
\end{alignat}

\begin{remark}
	\emph{The optimal locations for passive, active, and hybrid IRS under both sum- and product-distance path loss laws are determined by the integrated path loss distance of the cascaded BS$\rightarrow$IRS$\rightarrow$UE channel, small-scale fading of each link, and network parameters, such as $P_{\rm T}$, $P_{\rm F}$, $\delta^{2}$, $\delta_{\rm F}^{2}$, and $\alpha$.}
\end{remark}

\subsection{An Opportunistic Association Policy for the Hybrid IRS}

Based on the obtained optimal placement, we propose a two-step integrated path loss distance-based opportunistic association policy for the hybrid IRS. Assuming that the hybrid IRSs are densely deployed throughout the entire area, and the optimally located hybrid IRSs on the trajectories for all $d_{\rm BIU}$ are available, whose Cartesian coordinates are given by Lemma \ref{lem:localOp}. The association procedures are defined as follows:
\begin{enumerate}
	\item Determine the hybrid IRSs with the minimum integrated path loss distance of the cascaded channel, $d_{\rm BIU, min}$, according to the practical scenario.
	\item Select the hybrid IRS located at the optimal locations along the trajectory defined by $d_{\rm BIU, min}$ based on the results in Lemma \ref{lem:localOp}.
\end{enumerate}

In practical scenarios, the locations of BSs and IRSs are typically static. Consequently, when a UE connects to the network, it can report its location information to the BSs, and then the BSs have the capability to generate a lookup table for the UE, presenting both the index and the integrated path loss distance pertaining to potential IRSs. As such, both the BS and UE can determine the IRS with minimum integrated path loss distance.

\section{Performance Analysis} \label{sec:specialCase}

When the serving IRS is positioned optimally, the average received SNR of hybrid IRS-aided communication is expressed as
\begin{equation}\label{eq:outP_hy}
	\mathbb{E}[{\rm SNR}] = \frac{P_{\rm T}d_{\rm BIU,min}^{-\alpha}(N_{\rm pas}+\eta^{**} N_{\rm act})^{2}}{N_{\rm act}(\eta^{**})^{2}\delta_{\rm F}^{2} (d_{\rm IU}^{**})^{-\alpha} +\delta^{2} },
\end{equation}
where the optimal amplification factor is obtained by substituting \eqref{eq:hybrid_IRS_opt} into (P1), as shown below
\begin{equation}\label{eq:eta_opt}
	\begin{split}
		(\eta^{**})^2  = \begin{cases}  \frac{P_{\rm F}}{N_{\rm act}\left(P_{\rm T}\rho^{-\alpha}d_{\rm BIU,min}^{-\frac{\alpha}{2}}+\delta_{\rm F}^{2}\right)}, & \text {for product},  \\ \frac{P_{\rm F}}{N_{\rm act}\left(P_{\rm T}\left(\frac{d_{\rm BIU,min}}{\varphi_{\rm S}+1}\right)^{-\alpha}+\delta_{\rm F}^{2}\right)}, & \text{for sum }. \end{cases}
	\end{split}
\end{equation}
From \eqref{eq:outP_hy} and \eqref{eq:eta_opt}, it is observed that the distinction between hybrid and passive IRS-aided communication analysis lies in the amplification factor and the thermal noise introduced by the active parts. In addition, with the assistance of the geometric models, the average received SNR can be expressed as a single variable function with respect to $d_{\rm BIU}$. 

However, the coupling between the numerator and denominator of SNR resulting from the active sub-surface still leads to intractable analysis. Due to page limits, we leave this into our future work, and focus on exploring benefits of the geometric models and opportunistic association policy introduced in this work. To this end, we study a special case in the subsequent analysis, where $N_{\rm act}=0$ and the hybrid IRS is transformed into a passive IRS.\footnote{We will investigate the network-level performance of active sub-surface aided systems in our upcoming research.}

It is noteworthy that the existing research on the analysis of passive IRS-aided networks is still constrained by the challenge of modeling the integrated IRS-aided cascaded channel \cite{lyu2021hybrid,10458985}. While the integrated small-scale fading has been accurately and tractably modeled as a mixture Gamma distribution in \cite{10458985}, the precise modeling of the integrated path loss distance remains unsolved. The geometric models examined in this work effectively address this issue and facilitate the modeling of the entire integrated IRS-aided cascaded channel, employing an opportunistic association policy. 

We study two cases in this section, namely communications assisted by a single IRS or multiple IRSs. We consider an edge-user as the typical UE and assume that the direct link between the typical UE and its serving BS is blocked due to a long transmission distance.

\subsection{Single IRS-Aided Network}
In this subsection, we will delve into single IRS-aided communication, where the signal received by the typical UE from the connected IRS is 
\begin{equation}
	y =  {\epsilon}d_{\rm BIU}^{-{\alpha}/{2}} {\mathbf{g}}_{{\rm BIU}}x +n_{0},
\end{equation} 
where $x$ is the transmitted signal with power $P_{\rm T}$. Then, the SNR is given by
\begin{equation}
	{\rm SNR} = \frac{P_{\rm T}N^2\mathbb{E}\left[G_{\rm BIU}\right]d_{\rm BIU}^{-\alpha}}{\delta^{2}}.
\end{equation}

 \subsubsection{Outage Probability} \label{sec:nonOutageZone}
The outage probability is given by
%\begin{small}
\begin{equation} \label{eq:outagrP_single}
	\mathcal{P}_{\rm out} =  \mathcal{P}\left( {\rm SNR} < \tau \right)= \mathcal{P}\left(d_{\rm BIU}>\left(  \frac{P_{\rm T} N^2 \mathbb{E}\left[G_{\rm BIU}\right] }{\delta^{2}\tau}\right)^{\frac{1}{\alpha}}\right).
\end{equation}
Then, by substituting the CDF expression of the integrated path loss distance under different path loss laws, \eqref{eq:outagrP_single} can be resolved.
When the sum-distance path loss law is adopted, the outage probability is given by
\begin{equation} 
	\mathcal{P}_{\rm out} = 1 - \frac{\pi}{4S(d_{\rm t})} \varrho\tau^{-\frac{1}{\alpha}}\sqrt{\varrho^{2}\tau^{-\frac{2}{\alpha}}-4c^{2}},
\end{equation}
where $\varrho=\left(  \frac{P_{\rm T}N^2 \mathbb{E}\left[G_{\rm BIU}\right]}{\delta^{2}}\right)^{\frac{1}{\alpha}}$.
When the product-distance path loss law is adopted, the outage probability is given by
%\begin{small}
\begin{equation}
	\mathcal{P}_{\rm out} =
	\begin{cases}
		1 - \frac{2c^2}{S(d_{\rm t})}\left[E\left(\varpi\right)-\left(1-\varpi^{2}\right)K\left(\varpi\right)\right],  & \text{for } {\rm e}_{\rm P} < 1 \\ 
		1 - \frac{2\varpi c^{2}}{S(d_{\rm t})}E\left(\varpi^{-1}\right), & \text{for } {\rm e}_{\rm P}\geq 1 
	\end{cases},
\end{equation}
%\end{small}
where $\varpi=\frac{\varrho}{c^2\tau^{\frac{1}{\alpha}}}$.

\subsubsection{Low-Rate Probability}
The low-rate probability is defined as the probability that the rate is lower than the threshold, $r_{0}$, which is given by
%\begin{small}
\begin{equation} \label{eq:rateP_single}
	\begin{split}
		\mathcal{P}_{\rm R} =  &	\mathcal{P} \left( \ln(1+{\rm SNR}) < r_{0} \right) \\ %= \mathbb{P}\left( \ln \left[ 1+ \frac{P_{\rm T}N^2}{\delta^{2}}\mathbb{E}\left[G_{\rm BIU}\right]d_{\rm BIU}^{-\alpha} \right] < r_{0} \right)  
		= & \mathcal{P}\left( d_{\rm BIU} > \left[\frac{P_{\rm T}N^2 \mathbb{E}\left[G_{\rm BIU}\right] }{\delta^{2}(e^{r_{0}}-1)} \right]^{\frac{1}{\alpha}} \right),
	\end{split}
\end{equation}
%\end{small}
which can be expressed by the CDF of the integrated path loss distance as well.

\subsubsection{Relationship Between Outage Probability and Low-Rate Probability}
By comparing \eqref{eq:outagrP_single} and \eqref{eq:rateP_single}, the low-rate probability can be expressed as a function of the outage probability
\begin{equation}\label{eq:outagePrateP}
	\mathcal{P}_{\rm R} = \overline{F}\left[\left(\frac{\tau}{e^{r_{0}}-1} \right)^{\frac{1}{\alpha}} \overline{F}^{-1}(\mathcal{P}_{\rm out}) \right],
\end{equation}
where $\overline{F}(\cdot)=1-F(\cdot)$, and $F(\cdot)$ is the CDF of the integrated path loss distance, given in \eqref{eq:cassiniCDF} and \eqref{eq:ellipseCDF}.

\subsection{Multiple IRSs-Aided Communication}

In this subsection, we study the multiple IRSs-aided communication, where one reference cascaded link is selected, and the other IRSs align their phase shift with the reference link to corporately enhance the communication between transceivers.

\subsubsection{Association Policy}
To ensure effective operation within networks assisted by multiple IRSs, it is crucial to establish a suitable association policy for the IRSs. 

As shown in Definition~\ref{def:IrsSignalZone}, the typical UE associates to its serving BS through IRSs positioned within the IRS effective zone for ${\rm UE}_{0}$ and ${\rm BS}_{0}$, as stated below.
\begin{equation}
	\begin{split}
		\text{${\rm UE}_{0}$ associates with ${\rm BS}_{0}$ through a group of IRSs } \{{\rm IRS}_{l}\} \\
		\Leftrightarrow  \{{\rm IRS}_{l}\} = \arg_{{{\rm IRS}_l}\in \mathcal{J}} d_{{\rm BIU},l}\leq d_{\rm BIU,max},	\notag
	\end{split}
\end{equation}
where $l=1,\cdots,L$, $L$ is the total number of the connected IRSs. As the IRSs are uniformly distributed throughout the whole area, the number of IRSs located within the IRS effective zone is given by
\begin{equation}
	L = \lambda_{\rm I} \cdot S(d_{\rm BIU,max}) = \lambda_{\rm I} \cdot S\left(\varrho\tau^{-\frac{1}{\alpha}}\right), 
\end{equation}
where $\lambda_{\rm I}={J}/{S(d_{\rm t})}$ is the density of IRSs, $S(d_{\rm BIU,max})$ is the area of the IRS effective zone under the threshold, $\tau$, which depends on the selected path loss model, as shown in \eqref{eq:areaOfCassiniOval} and \eqref{eq:areaOfEllipse}. 

Thus, the received signal at the typical UE is given by 
%\begin{small}
\begin{equation}
	y = \sum_{l=1}^{L} {h}_{{\rm BIU},l}  x + n_{0},
\end{equation}
%\end{small}
where $ {h}_{{\rm BIU},l}, l=1,\cdots,L,$ is the channel of the $l$-th cascaded BS$\rightarrow$${\rm IRS}_{l}$$\rightarrow$UE link. 

When the serving IRSs are located at the same equal gain trajectory, the average received signal power is given by
%\begin{small}
\begin{equation} \label{eq:sameLineIrs}
	\mathbb{E}[P_{\rm r}]  =  L^{2}N^{2} P_{\rm T} \epsilon^{2} d_{\rm BIU}^{-{\alpha}}\mathbb{E}[G_{\rm BIU}]. 
\end{equation}
%\end{small}
Then, the average received SNR is given by
%\begin{small}
\begin{equation}\label{eq:SNR_co}
	\mathbb{E}[{\rm SNR}]  =   L^{2} \mathbb{E}[{\rm SNR}_{l}], \quad \mathbb{E}[{\rm SNR}_{l}] = N^{2} \frac{P_{\rm T} \epsilon^{2} d_{\rm BIU}^{-{\alpha}}\mathbb{E}[G_{\rm BIU}]}{\delta^{2}}.\notag
\end{equation}
%\end{small}
It is observed that the phase shift alignment among multiple associated IRSs provides a SNR gain in order of $O(L^{2}N^{2})$, which is consistent with a centralized IRS with $LN$ reflecting elements.

\subsubsection{Channel Statistics}

To evaluate the performance metrics, we focus on the mixture channel coefficient $|{h}_{\rm S}| \triangleq |\sum_{l=1}^{L} {h}_{{\rm BIU},l}|$. 
\begin{proposition}\label{prop:channelModel}
	Given each single link follows generalized fading distribution, the amplitude of each cascaded channel, $|{ h}_{{\rm BIU},l}|$, can be modeled as a mixture Nakagami-$m$ distribution, whose PDF is given by
	\begin{equation} \label{append:mixtureNaka}
		f_{\rm X}(x) = \sum_{i=1}^{I} 2\varepsilon_{i}x^{2\beta_{i}-1}e^{-\xi_{i}x^{2}},
	\end{equation} 
	where $ I $ is the number of Nakagami-$m$ terms, which is set as $17$ to achieve an error less than $10^{-5}$, $ \{ \varepsilon_{i}, \beta_{i}, \xi_{i} \}$ is the parameters of each term \cite{abramowitz1964handbook}. 
\end{proposition}
\begin{proof}
	See Appendix~\ref{append:mixNaka}.
\end{proof}

\begin{lemma}\label{lem:coefficient}
		The amplitude of the mixture channel coefficient can be modeled as a Gaussian distributed random variable, denoted as $|{ h}_{\rm S}|\sim \mathcal{N}(\mu,\delta_{\rm S})$, where $\mu$ and $\delta_{\rm S}$ are given by 
	\begin{equation}
		\mu  =  \sum_{l=1}^{L} \mu_{l}, \quad \delta_{\rm S}^{2}  = L - \sum_{l=0}^{L} \mu_{l} ^{2} ,
	\end{equation}
	%where 
	\begin{equation}\label{eq:mu_l}
		\mu_{l} =  \sum_{i=1}^{I} \varepsilon_{i}^{(l)}\left( \beta_{i}^{(l)}-\frac{1}{2}\right) ! {\xi_{i}^{(l)}}^{-\beta_{i}-\frac{1}{2}} , l = 1,\cdots,L,
	\end{equation}
	%$ \mu_{l}=\sum_{i=1}^{I} \varepsilon_{i}(\beta_{i}-\frac{1}{2})! \xi_{i}^{-\beta_{i}-\frac{1}{2}} $,
        and $\{ \varepsilon_{i}^{(l)}, \beta_{i}^{(l)}, \xi_{i}^{(l)} \}$ is the distribution parameters of the channel coefficient of the $l$-th serving IRS-aided cascaded link ($\text{BS}_{0}\rightarrow\text{IRS}_{l}\rightarrow\text{UE}_{0}$) \cite{abramowitz1964handbook}.
\end{lemma}
\begin{proof}
See Appendix~\ref{append:CLT_h}.
\end{proof}

\begin{proposition}
	When the serving IRSs are located on the equal trajectory with identical $d_{\rm BIU}$, and each IRS-aided channel suffers the same small-scale fading, the received signal power can be modeled as $|{h}_{\rm S}|\sim \mathcal{N}(\mu,\delta_{\rm S}^2)$, where $\mu$ and $\delta_{\rm S}^2$ are 
	\begin{equation}
		\mu  =  {L} \mu_{l}, \quad \delta_{\rm S}^{2}  =  L - L \mu_{l} ^{2}.
	\end{equation} 
 \end{proposition}

\subsubsection{Received Signal Power}
As the serving IRSs align phase shifts with a reference link, the received signal power is given by
%\begin{small}
\begin{equation}
	{ S} = P_{\rm T}\left| \sum_{l=1}^{L} { h}_{{\rm BIU},l}\right|^2 . 
\end{equation}
%\end{small}
Then, we model the received signal power as a Gamma distribution with the assistance of the second and fourth moments of amplitude of the channel coefficients, denoted as ${ S} \sim {\rm Gamma}(k,\theta)$ with parameters \cite{lyu2021hybrid}
\begin{equation} \label{eq:signalDist}
	k = \frac{\left( \mathbb{E}[|{ h}_{\rm S}|^2] \right)^2}{\mathbb{E}[|{ h}_{\rm S}|^4] -\left( \mathbb{E}[|{h}_{\rm S}|^2]  \right)^{2}  }, \quad \theta = \frac{\mathbb{E}[|{ h}_{\rm S}|^4] -\left( \mathbb{E}[|{ h}_{\rm S}|^2]  \right)^{2}  }{ \mathbb{E}[|{ h}_{\rm S}|^2]} ,
\end{equation}
where $ \mathbb{E}[|{ h}_{\rm S}|^2] =\delta_{\rm S}^{2}+\mu^{2} $, and $ \mathbb{E}[|{h}_{\rm S}|^4] = 3\delta_{\rm S}^{4}+6\delta_{\rm S}^2\mu^2+\mu^4 $.

As the received signal power is modeled as a Gamma distribution, denoted as $ {S}\sim {\rm Gamma}(k,\theta)$, the performance be evaluated using the tools in \cite{hamdi2007useful} since all the performance metrics can be expressed as functions of SNR \cite{10458985,chun2017comprehensive}.%,7890358}.

\subsubsection{Spectral Efficiency}

Spectral efficiency is defined as %\cite{jo2012heterogeneous}
\begin{equation} \label{spectralEfficiency}
	\mathcal{R} = \mathbb{E}[\ln(1+\rm SNR)].
\end{equation}
By utilizing the distribution parameters of the signal power in \eqref{eq:signalDist}, the spectral efficiency can be derived by utilizing the analysis framework proposed in \cite{hamdi2007useful}.
\begin{equation}
	\begin{split}
		\mathcal{R} =  \int_{0}^{\infty} { {\frac{1}{z}\left( 1-\frac{1}{(1+z)^{k}} \right) } {e^{-\frac{\delta^{2} z}{\theta}}} }  \mathrm{d}z.
	\end{split} 
\end{equation}

\subsubsection{Moments of SNR}

The moments of the SNR $\mathbb{E}[{\rm SNR}^{r}]$ can be evaluated as follows \cite{chun2017comprehensive}
\begin{equation}
	\begin{split}
		\mathbb{E}[{\rm SNR}^{r}] = \int_{0}^{\infty} { { \frac{\Gamma(k+r)}{\Gamma(k)\Gamma(r)}z^{r-1} } e^{-\frac{\delta^{2} }{\theta}z}}   \mathrm{d}z  .
	\end{split} 
\end{equation}

\section{Numerical Results}\label{sec:numericalR}

In this section, Monte-Carlo (MC) simulation with $10^{6}$ iterations is carried out using MATLAB to verify the analysis and compare the performance under different system setups. The simulation setup is as follows, if not specified otherwise. The path loss exponent is $\alpha=3$, the transmit power of the BS is $P_{\rm T}=20$ dBm, the amplification power of the active reflecting elements is $P_{\rm F}=5$ dBm, the noise power at the receiver is $\delta^{2}=-90$ dBm, the thermal noise power generated by the active reflecting element is $\delta_{\rm F}^{2}=-80$ dBm, the half of $d_{\rm BU}$ is $c=30$ m, the number of IRS elements are $N_{\rm pas}=300$, $N_{\rm act}=40$, the carrier frequency is 2 GHz \cite{amato2018tunneling,kang2023active}.

%% figures
\begin{figure*}[t!]
	\centering
	\subfigure[Product-distance path loss law.]{
		\begin{minipage}[t]{0.47\textwidth}
			\includegraphics[width=\linewidth]{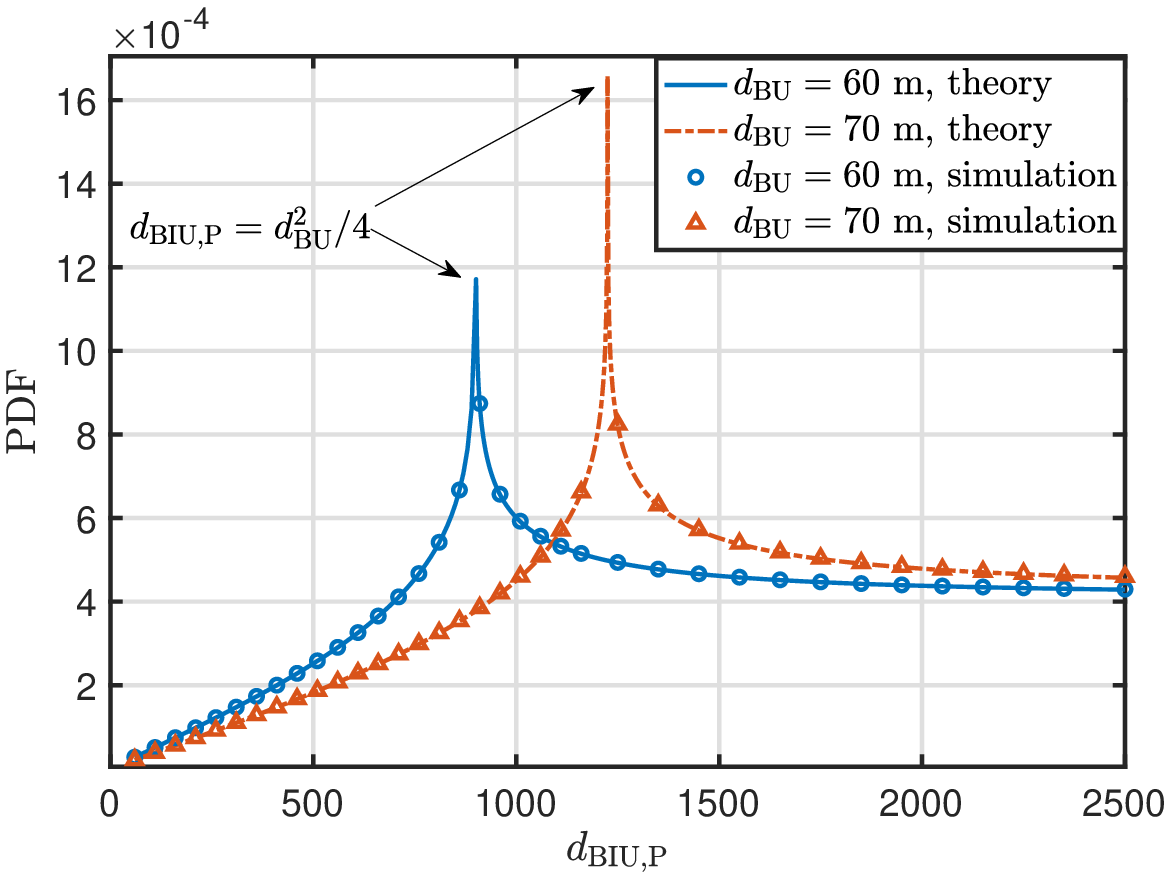}
		\end{minipage}
		\label{fig:PDF_cassiniOval}
	}
	\subfigure[Sum-distance path loss law.]{
		\begin{minipage}[t]{0.47\textwidth}
			\includegraphics[width=\linewidth]{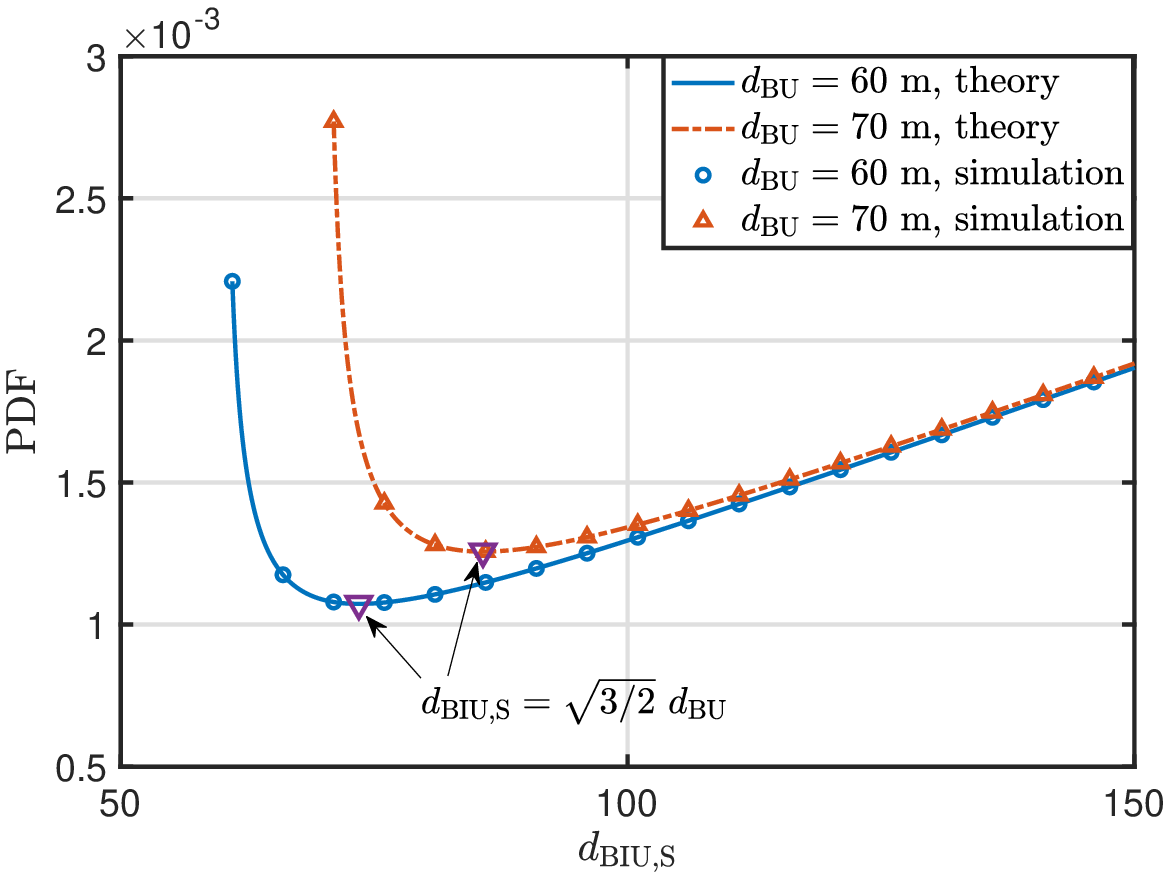}
		\end{minipage}
		\label{fig:PDF_ellipse}
	}
	\caption{PDF of the integrated path loss distance, $d_{\rm BIU}$.}
	\label{fig:PDF}
\end{figure*}

\subsection{Distribution of the Integrated Path Loss Distance}

We first display the PDF of the integrated path loss distance, $d_{\rm BIU}$, with different $d_{\rm BU}$ for both product- and sum-distance path loss laws in Fig.~\ref{fig:PDF}, thereby validating the findings in Lemma~\ref{lem:cassiniOvalPCDF} and Lemma~\ref{lem:ellipsePCDF}.

It has been observed that the PDF of the integrated path loss distance scales with the link distance between BS and UE. The results align with the geometric models depicted in Fig.~\ref{fig:IRStrjectory}, indicating that the distribution of $d_{\rm BIU}$ is dependent on $d_{\rm BU}$. Consequently, even in instances where the direct link is obstructed, a BS in closer proximity of the typical UE, as indicated by a smaller $d_{\rm BU}$, may have a higher probability of delivering superior performance.

In Fig.~\ref{fig:PDF_cassiniOval}, for product-distance path loss law, it is observed that the PDF exhibits an initial increase followed by a subsequent decrease and the demarcation point is $d_{\rm BIU,P}=d_{\rm BU}^2/4$, i.e., ${\rm e}_{\rm P}=1$. On the other hand, an opposite trend is observed in the sum-distance path loss law and the demarcation point is $d_{\rm BIU,S}=\sqrt{{3}/{2}}~d_{\rm BU}$, i.e., ${\rm e}_{\rm S}= \sqrt{{2}/{3}}$, as displayed in Fig.~\ref{fig:PDF_ellipse}. By comparing Figs.~\ref{fig:PDF_cassiniOval} and \ref{fig:PDF_ellipse}, it can be observed that the system suffers severer degradation when employing the product-distance path loss law.

\subsection{Achievable Rate Under Different Path Loss Laws}% of the hybrid IRS-aided communication }

The three-dimensional (3-D) simulation of the achievable rate in the IRS-aided communication is presented in Fig.~\ref{fig:achievableR_3D}, showcasing the impact of different path loss laws for passive/hybrid IRS. % We have the following observations. 

As shown in Figs.~\ref{fig:achievableR_pro_3D_pas} and \ref{fig:achievableR_sum_3D_pas}, when the IRS is fully passive, the performance has strong geometric symmetry, which aligns with the geometric models of the integrated path loss for the cascaded channels introduced in Section \ref{sec:geo}. This is expected since the received SNR of passive IRS-aided communication is proportional to the channel gain of the integrated cascaded channel.

However, the adoption of the hybrid IRS breaks the symmetry between the UE side and BS side, owing to the location preference of the active IRS sub-surface, as shown in Figs.~\ref{fig:achievableR_pro_3D_hy} and \ref{fig:achievableR_sum_3D_hy}. Furthermore, the optimal locations of the hybrid IRS depend on the integrated path loss distance of the cascaded channel and the system parameters, such as the transmit power of the serving BS, $P_{\rm T}$, the amplification power of the active reflecting elements, $P_{\rm F}$, the power of the thermal noise introduced by the active parts, $\delta_{\rm F}^{2}$, the thermal noise at the typical UE, $\delta^{2}$, and the path loss exponent, $\alpha$. 

Moreover, compared to the product-distance path loss law, the impact of IRS locations is relatively insignificant when adopting the sum-distance path loss law, as shown in Fig.~\ref{fig:achievableR_sum_3D_hy}. For example, in the case of passive IRS, the achievable rate gap that can be achieved across the entire area is about 15 nats/sec/Hz for the product-distance path loss law, while it is around 3 nats/sec/Hz for the sum-distance path loss law. %6) In Fig. \ref{fig:achievableR_sum_3D_hy}, the minimum value is achieved when the IRS is located near the location of the typical BS.

\subsection{Nearest and Opportunistic IRS Association Policy}

In Fig.~\ref{fig:perfComp}, we present the average received SNR from IRS-aided channels versus the IRS density for both product-/sum-distance path loss laws. The 'o' and 'n' represent opportunistic and nearest association policy, respectively. The parameters for the simulations are: $\alpha=4$, $\eta \leq 14$ dB, the total deployment budget is $500$, and the deployment cost of active(passive) reflecting element is 5. Specifically, the number of IRS element is $N=500$ for passive IRS, while the number of IRS element is $N=100$,  for active IRS \cite{amato2018tunneling,kang2023active}. We observed that the performance varies across different configurations, such as association policies, e.g., nearest and opportunistic association, and types of IRS, i.e., passive/active/hybrid IRS. It is noteworthy that the hybrid IRS offers better performance than passive/active IRS with the same power and deployment budget for both product-/sum-distance path loss laws, under the mentioned setup.

Moreover, there are big performance gaps between the two association policies for all three types of IRSs. These results verifies that even for passive IRS, the nearest IRS might be a sub-optimal option. Although the nearest association provides a similar trend for the passive IRS-aided communication, there is still a big performance gap between these two association policies, which can be explained by the geometric models.

However, for the active and hybrid IRS, these two association policy behaves reversely. If nearest association policy is adopted, the performance of active/hybrid IRS decreases with increasing density of IRS. On the contrary, the performance increases with IRS density when opportunistic association policy is adopted. This is reasonable because the active IRS is supposed to be deployed at a certain distance from the typical UE \cite{you2021wireless,li2022active}. When the density of IRS is increasing, the connect distance between the UE and its nearest IRS is decreased, which would degrade the performance enhancement of the active IRS due to the instantaneous amplification and thermal noise introduced by the active parts.

\subsection{Outage Probability and Rate Probability}

The relationship between outage probability and rate probability, defined as one minus low-rate probability, is illustrated in Fig. \ref{fig:outagePrateP} for two path loss models. Generally, the curves for the sum-distance model are more sensitive to changes in the $\frac{\tau}{{e^{r_{0}}-1}}$, which is an essential factor in \eqref{eq:outagePrateP}. The curves fall into three distinct categories: a) the outage probability and rate probability are equal given $\frac{\tau}{{e^{r_{0}}-1}}=1$, b) the outage probability exceeds the rate probability when $\frac{\tau}{{e^{r_{0}}-1}}<1$, and c) the outage probability is lower than the rate probability when $\frac{\tau}{{e^{r_{0}}-1}}>1$. In a single cell, the outage probability for the same rate probability increases with $\frac{\tau}{{e^{r_{0}}-1}}$. The effect of path loss exponents is presented for both sum- and product-distance path loss laws. The results indicate that as the path loss exponent increases, the benefit of IRSs is limited due to severe path loss of the cascaded channels. 

\subsection{Channel Modeling of the Mixture Channel of Multiple IRSs}

In Fig.~\ref{fig:channel}, we present the distribution for small-scale channel coefficients of the mixture channel composed by multiple serving IRSs. It is observed that the approximated channel distribution proposed in Lemma~\ref{lem:coefficient} closely aligns with the empirical results. Furthermore, as $L$ increases, the PDF of the channel coefficient exhibits reduced kurtosis, and the mixture channel becomes stronger and provides a greater channel gain. {Thus, the utilization of multiple IRSs in communication enables greater flexibility in deploying the IRSs, resulting in comparable performance to that of a single serving IRS.}

\begin{figure*}[t!]
	\centering
	\subfigure[Product-distance path loss law (Passive IRS).]{
		\begin{minipage}[t]{0.47\textwidth}
			\includegraphics[width=\linewidth]{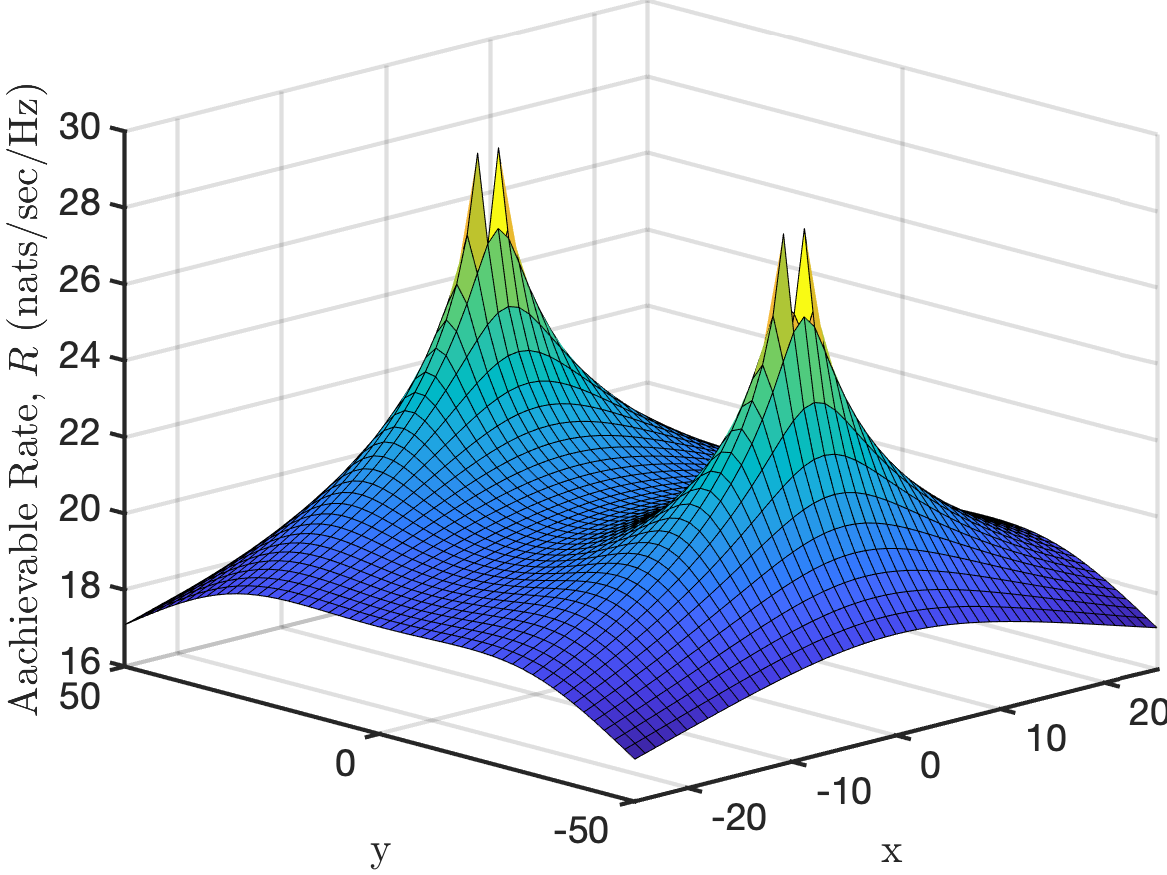}
		\end{minipage}
		\label{fig:achievableR_pro_3D_pas}
	}
	\hspace{2pt}
	\subfigure[Sum-distance path loss law (Passive IRS).]{
		\begin{minipage}[t]{0.47\textwidth}
			\includegraphics[width=\linewidth]{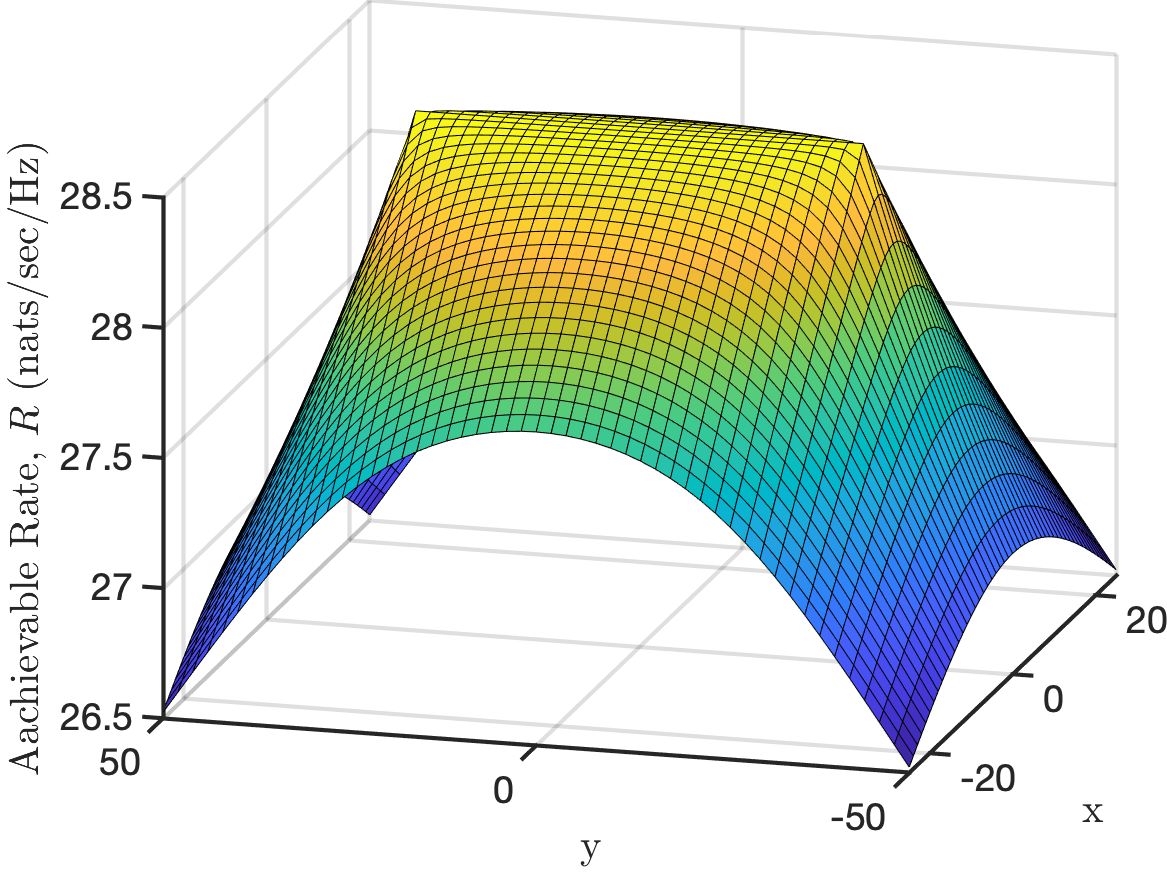}
		\end{minipage}
		\label{fig:achievableR_sum_3D_pas}
	}
	\hspace{2pt}
	\centering
	\subfigure[Product-distance path loss law (Hybrid IRS).]{
		\begin{minipage}[t]{0.47\textwidth}
			\includegraphics[width=\linewidth]{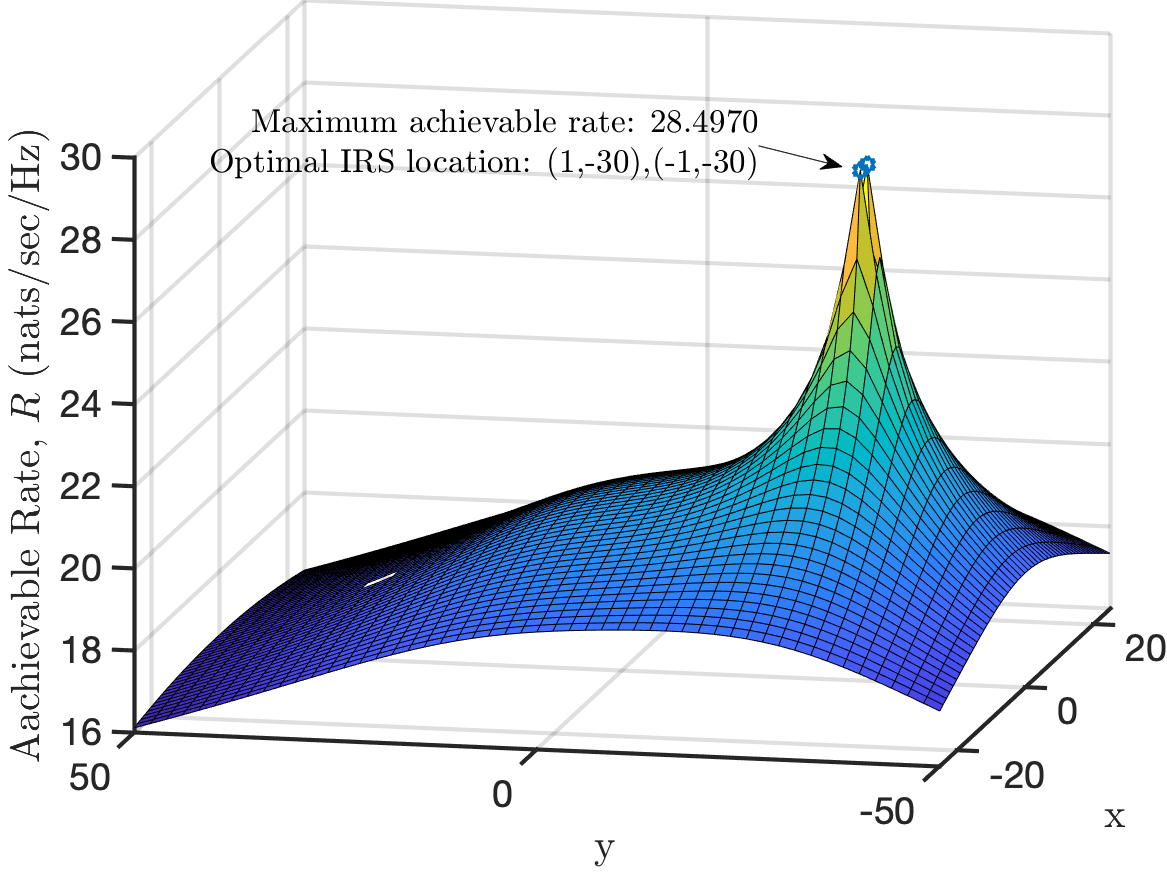}
		\end{minipage}
		\label{fig:achievableR_pro_3D_hy}
	}
	\hspace{2pt}
	\subfigure[Sum-distance path loss law (Hybrid IRS).]{
		\begin{minipage}[t]{0.47\textwidth}
			\includegraphics[width=\linewidth]{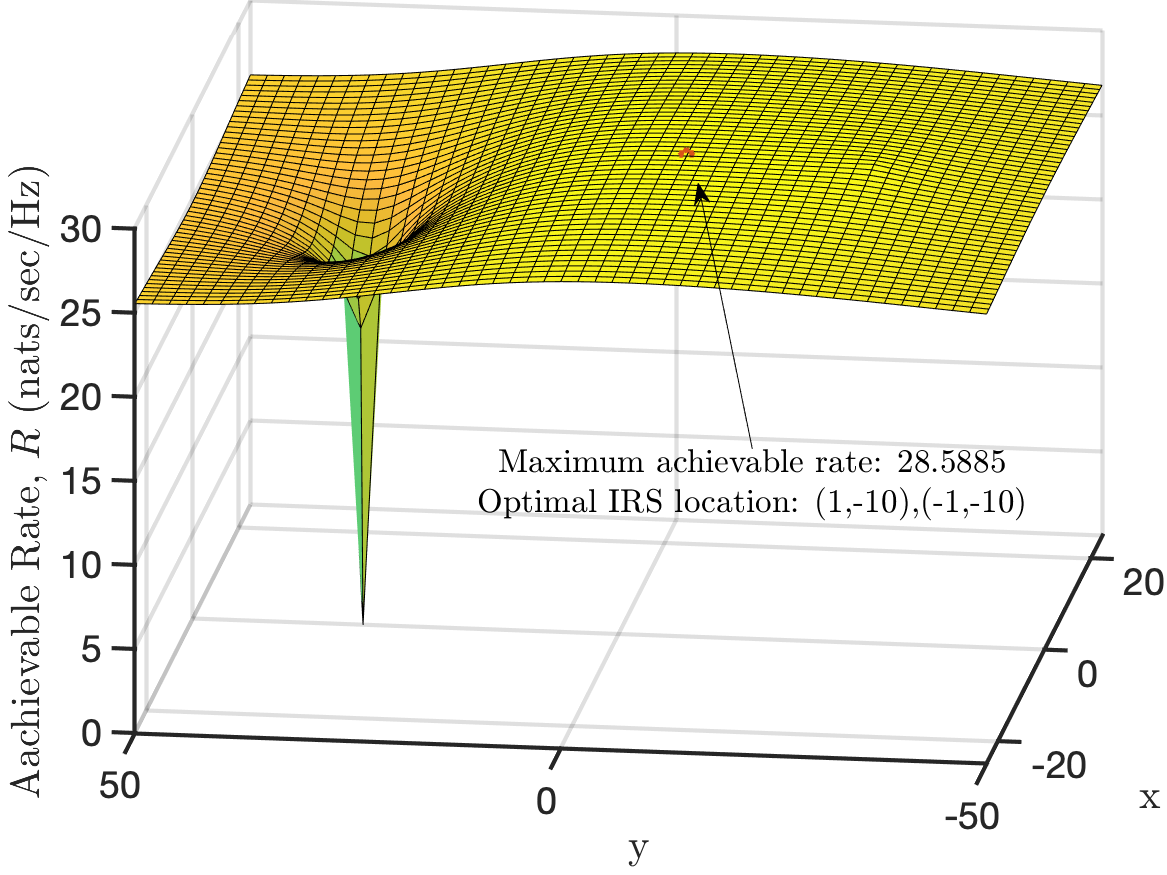}
		\end{minipage}
		\label{fig:achievableR_sum_3D_hy}
	}
	\caption{Achievable rate across the whole area.}
	\label{fig:achievableR_3D}
\end{figure*}

\begin{figure*}[t!]
	\centering
	\subfigure[Product-distance path loss law.]{
		\begin{minipage}[t]{0.47\textwidth}
			\includegraphics[width=\linewidth]{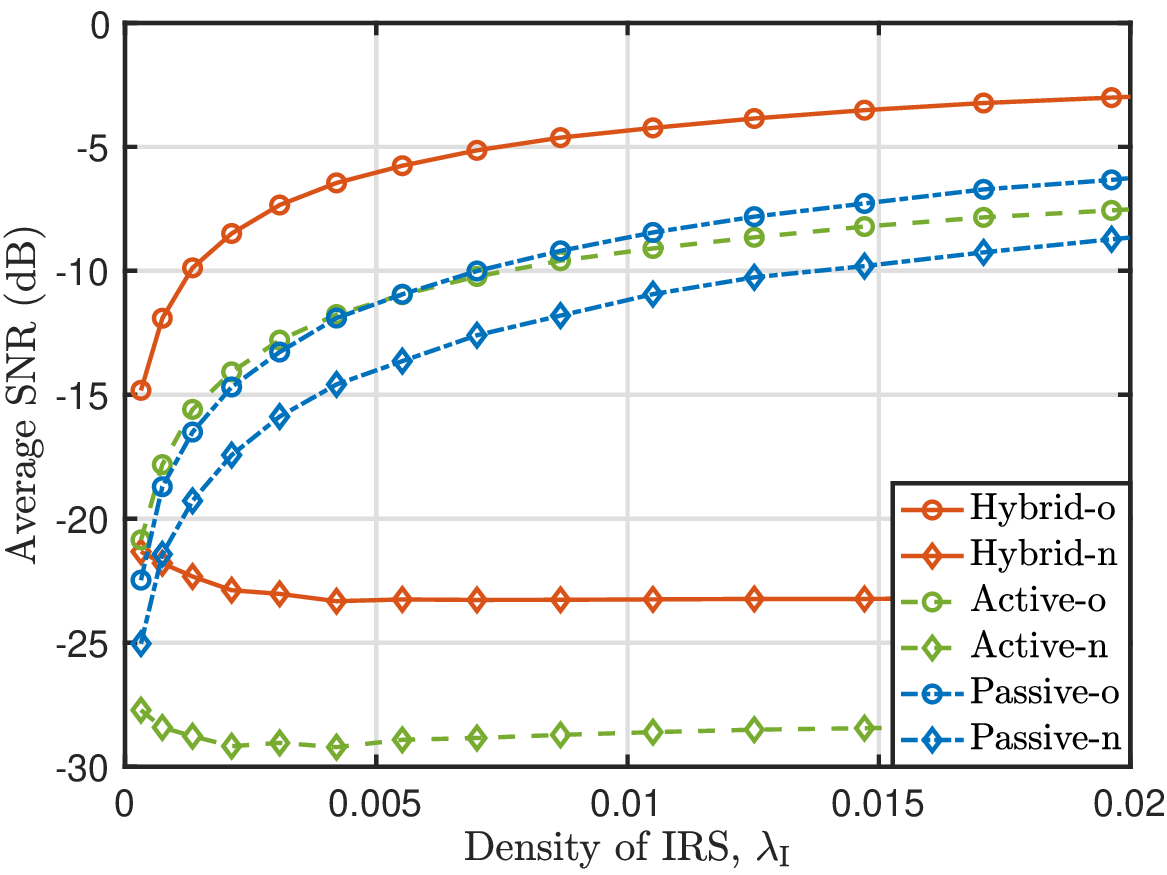}
		\end{minipage}
		\label{fig:perfComp_pro}
	}
	\hspace{2pt}
	\subfigure[Sum-distance path loss law.]{
		\begin{minipage}[t]{0.47\textwidth}
			\includegraphics[width=\linewidth]{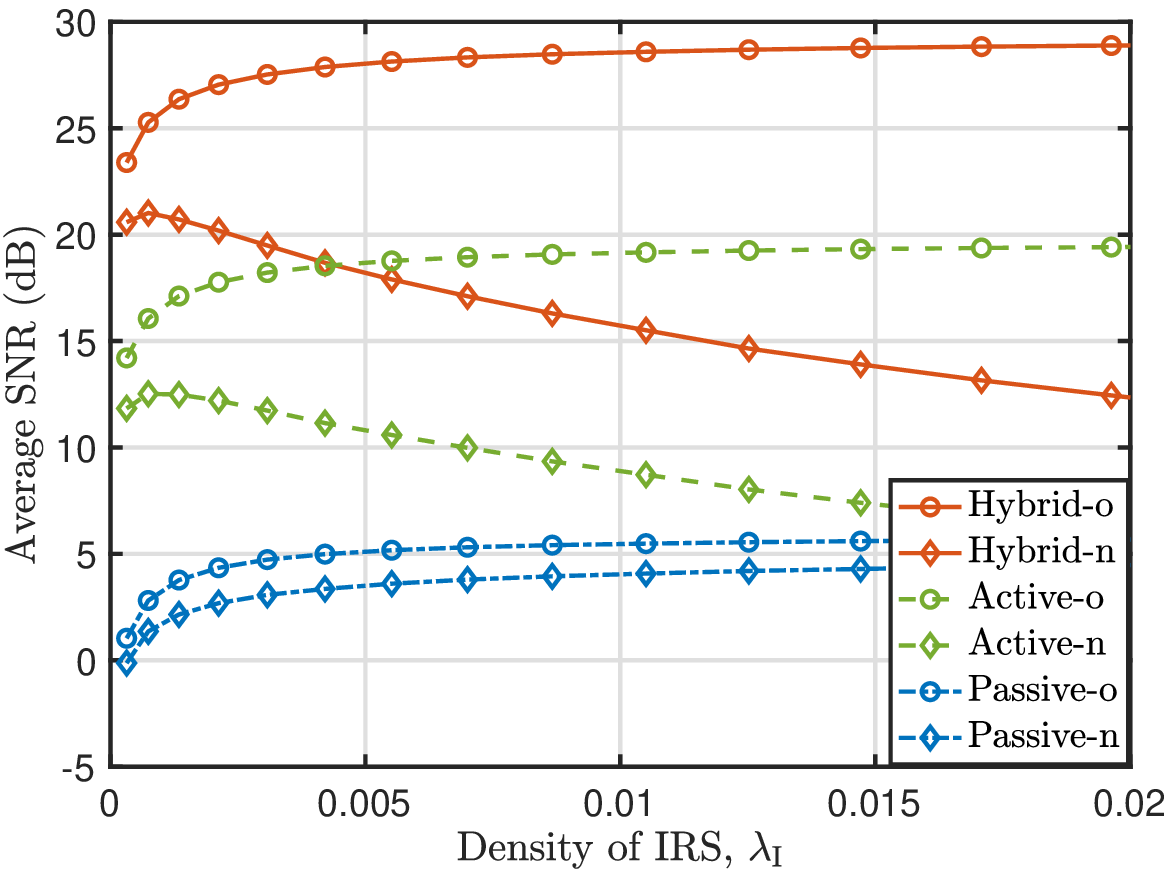}
		\end{minipage}
		\label{fig:perfComp_sum}
	}
	\caption{Performance comparison between nearest and opportunistic association policies across various types of IRS with same power and deployment budget.}
	\label{fig:perfComp}
\end{figure*}

\begin{figure*}[t!]
	\centering
	\subfigure[Product-distance path loss law.]{
		\begin{minipage}[t]{0.47\textwidth}
			\includegraphics[width=\linewidth]{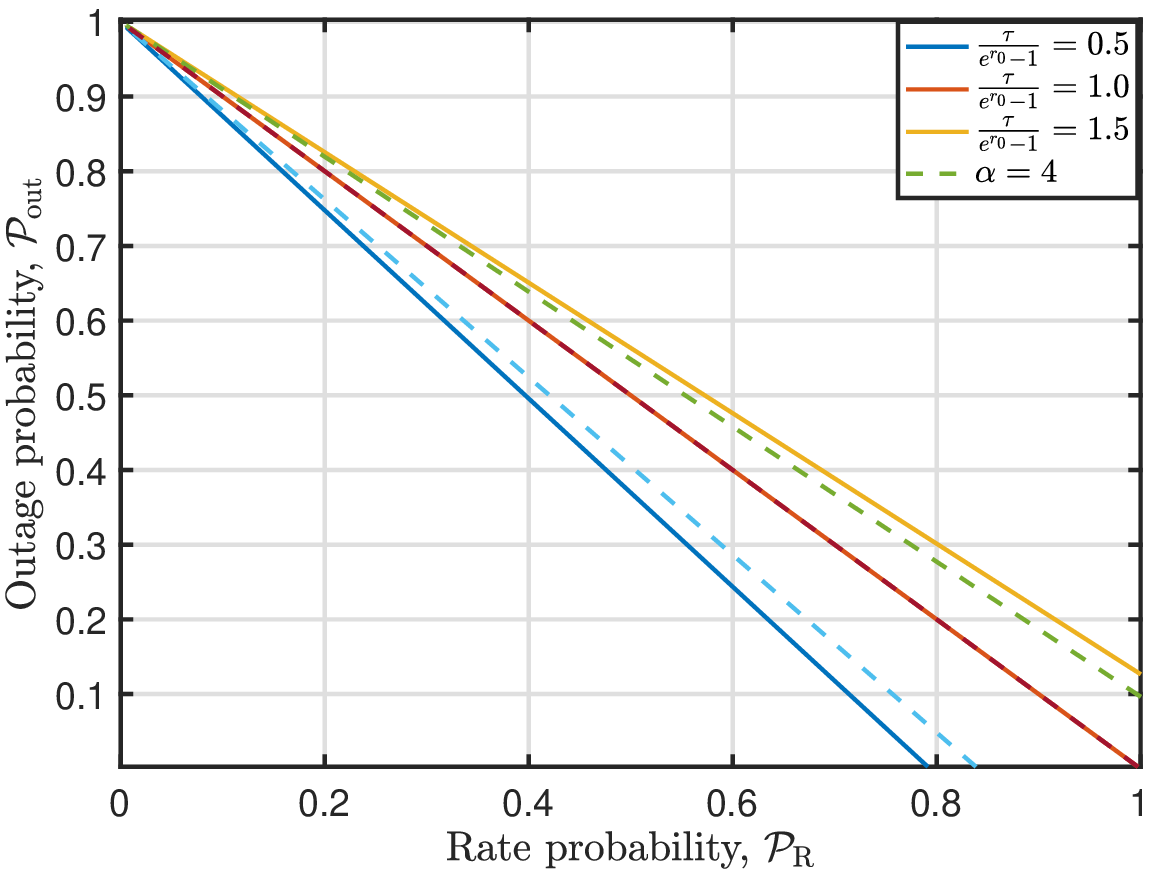}
		\end{minipage}
		\label{fig:outagePrateP_product_singleCell}
	}
	\hspace{2pt}
	\subfigure[Sum-distance path loss law.]{
		\begin{minipage}[t]{0.47\textwidth}
			\includegraphics[width=\linewidth]{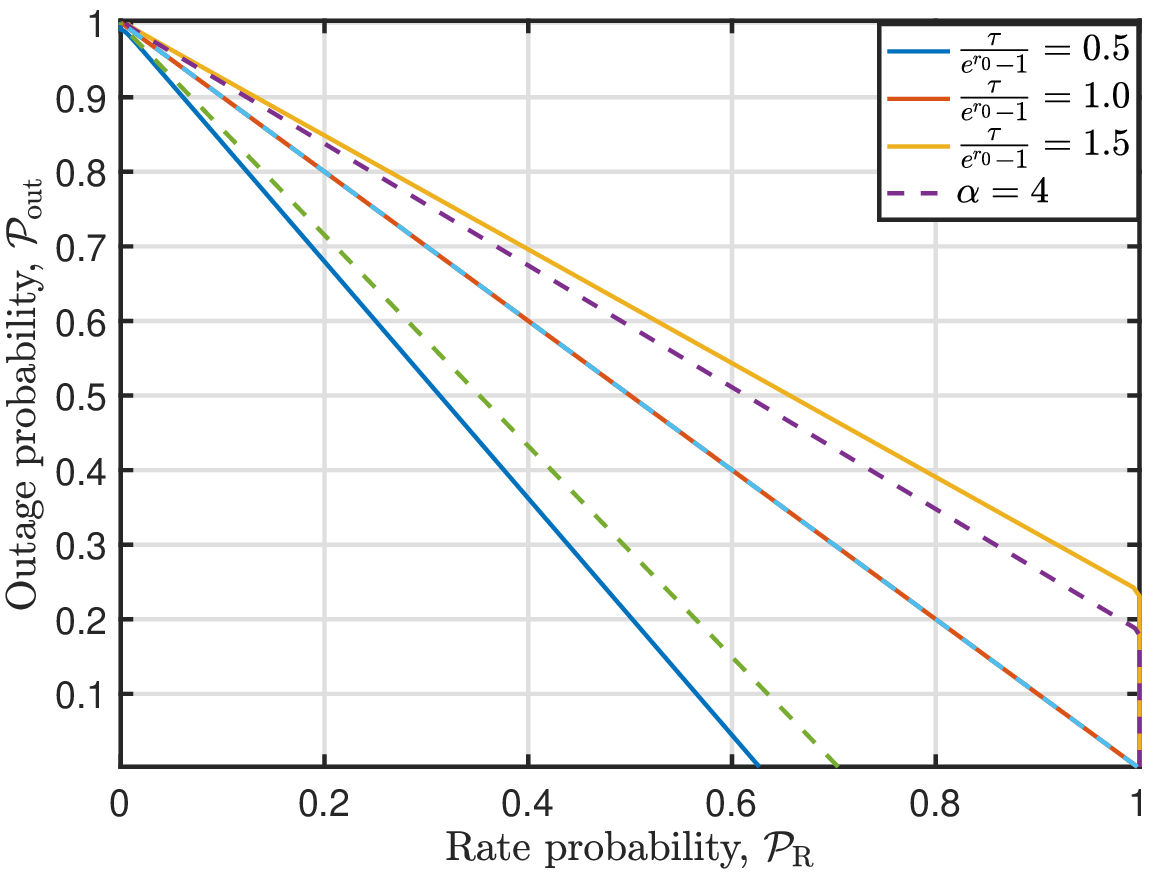}
		\end{minipage}
		\label{fig:outagePrateP_sum_singleCell}
	}
	\caption{Relationship between the outage probability and rate probability.}
	\label{fig:outagePrateP}
\end{figure*}

\begin{figure}[t!]
	\centering
        \includegraphics[width=0.94\linewidth]{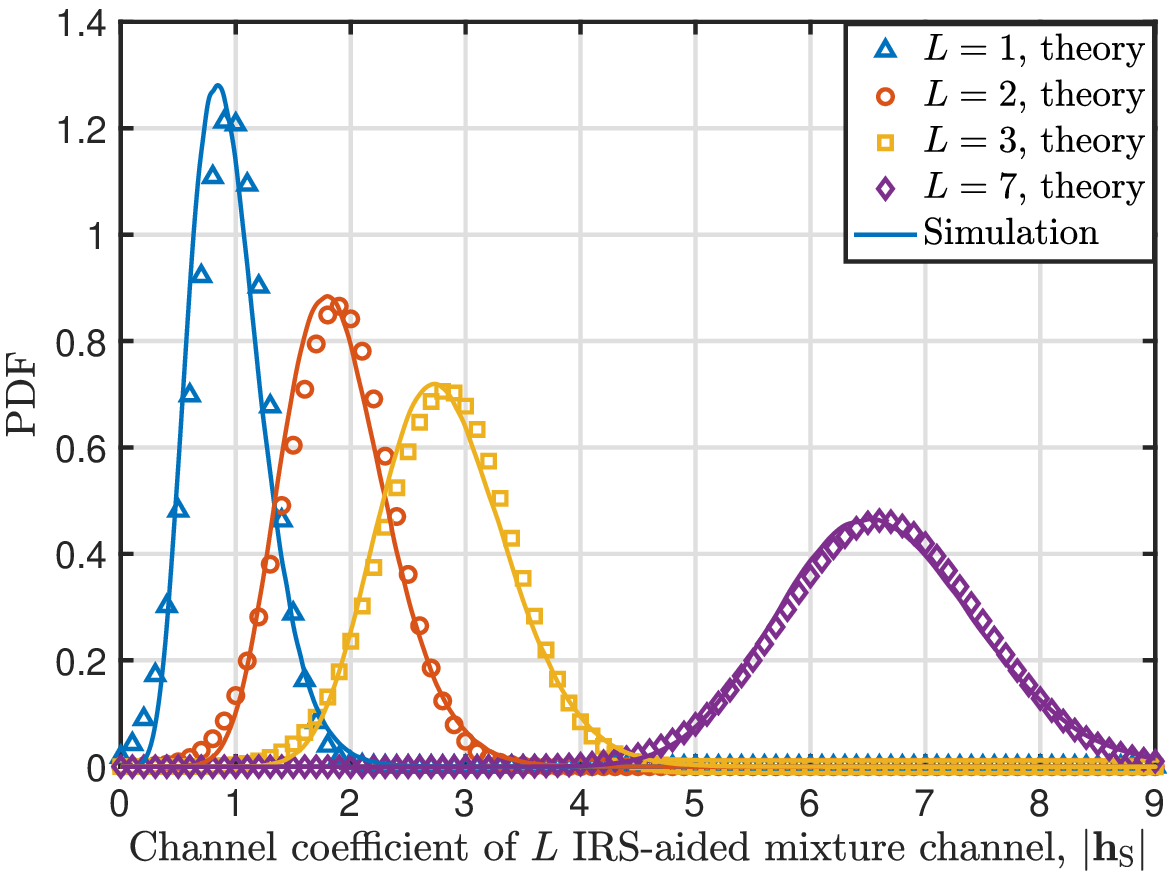}
        \caption{Verification of the approximated channel model for the mixture channel of multiple IRSs-aided scenarios.}
        \label{fig:channel}
\end{figure}

\section{Conclusion} \label{sec:conclusion}

In this work, we studied the geometric models to characterize the cascaded channel’s integrated path loss distance in IRS-aided communications. We derived the statistical distribution of the integrated path loss distance and applied these geometric models to determine the hybrid IRS’s optimal locations, which encompass both passive and active IRS as special cases. Our findings indicate that the hybrid IRS’s optimal placement depends on the integrated path loss distance of the cascaded channel and the network parameters. This solution can be utilized to develop an opportunistic association policy for the IRS and strategic deployment planning. The geometric models offer a new perspective for analyzing IRS-aided networks’ performance, allowing accurate statistical modeling of the integrated cascaded channel. Furthermore, we studied several scenarios to explore the potential implementations of the geometric model.

\appendices

\section{} \label{append:theoremCassiniOval}
In this appendix, we provide a proof of Theorem \ref{theorem:cassiniOval}.
The area calculation of the Cassini oval is divided into two cases: ${\rm e}_{\rm P}\geq 1$ and ${\rm e}_{\rm P}<1$, as shown below, where ${\rm e}_{\rm P}=\frac{\sqrt{d_{\rm BIU,P}}}{c}$.  

{ \bf Case 1:} When ${\rm e}_{\rm P}\geq1$, the curve is continuous, and the polar radius of a Cassini Oval is %given by
\begin{equation}
	r^2 = c^2\left[ \cos(2\theta)+\sqrt{{\rm e}_{\rm P}^4 -\sin^{2}(2\theta)} \right],
\end{equation}
and then, the area is given by
%\begin{small}
\begin{equation}
	\begin{split}
		S = &~ 4\int_{0}^{\frac{\pi}{2}} \frac{1}{2} r^2 {\rm d} \theta \\= &~
		2c^2\left[  \int_{0}^{\frac{\pi}{2}} \cos(2\theta) {\rm d} \theta + \int_{0}^{\frac{\pi}{2}} \sqrt{{\rm e}_{\rm P}^4 -\sin^{2}(2\theta)} {\rm d} \theta \right]  \\  = &~ 2c^2 \left[ \int_{0}^{\frac{\pi}{4}} \sqrt{{\rm e}_{\rm P}^4 -\sin^{2}(2\theta)} {\rm d} \theta + \int_{\frac{\pi}{4}}^{\frac{\pi}{2}} \sqrt{{\rm e}_{\rm P}^4 -\sin^{2}(2\theta)} {\rm d} \theta \right]\\  = &~2c^2 \int_{0}^{\frac{\pi}{2}} {\rm e}_{\rm P}^{2} \sqrt{1-{\rm e}_{\rm P}^{-2}\sin^{2}(\theta)}{\rm d}\theta = 2d_{\rm BIU,P}E({\rm e}_{\rm P}^{-2}), \notag
	\end{split}
\end{equation}
%\end{small}
where $E(k)=\int_{0}^{\frac{\pi}{2}}\sqrt{1-k^2\sin^{2}\theta}{\rm d} \theta$ is the complete elliptic integral of the second kind \cite{abramowitz1964handbook}.

{ \bf Case 2:} When ${\rm e}_{\rm P}<1$, the curve is discontinuous, and the polar radius of a Cassini Oval is given by
\begin{equation}
	r_{1,2}^2 = c^2\left[ \cos(2\theta)\pm\sqrt{{\rm e}_{\rm P}^4 -\sin^{2}(2\theta)} \right].
\end{equation}
The range of $\theta$ is given by ${\rm e}_{\rm P}^4=\sin^{2}(2\theta)$, which results in $\theta\in[-\theta_{0},\theta_{0}]$ and $\theta_{0}=\frac{1}{2}\arcsin({\rm e}_{\rm P}^2)$.
Then, the area is 
\begin{equation}
	S =  4\int_{0}^{\theta_{0}} \frac{1}{2} (r_{1}^2-r_{2}^2) {\rm d} \theta ,
\end{equation}
which can be further derived as
%\begin{small}
\begin{equation}
	\begin{split}
		S = &~ 4c^2\int_{0}^{\theta_{0}} \sqrt{{\rm e}_{\rm P}^4 -\sin^{2}(2\theta)} {\rm d} \theta\\   = &~ 2c^2\int_{0}^{2\theta_{0}} \sqrt{{\rm e}_{\rm P}^4 -\sin^{2}(\theta)} {\rm d} \theta \\ = &~ 2d_{\rm BIU,P}E\left[ \sin^{-1}({\rm e}_{\rm P}^2),{\rm e}_{\rm P}^{-2}\right] \\  \overset{(a)}{=} &~ 2d_{\rm BIU,P}\int_{0}^{1} \frac{1-x^2}{\sqrt{1-{\rm e}_{\rm P}^4x^2}\sqrt{1-x^2}} {\rm e}_{\rm P}^2 {\rm d}x \\ = &~ 2c^2{\rm e}_{\rm P}^4 \left[ K({\rm e}_{\rm P}^2)-\int_{0}^{1} \frac{x^2}{\sqrt{1-{\rm e}_{\rm P}^4x^2}\sqrt{1-x^2}} {\rm d}x \right] \\ = &~ 2c^2 \left[ E({\rm e}_{\rm P}^2) -(1-{\rm e}_{\rm P}^{4})K({\rm e}_{\rm P}^2) \right], \notag
	\end{split}
\end{equation}
%\end{small}
where step (a) follows a substitution of $x={\rm e}_{\rm P}^{-2}t $, $E(\phi,k)=\int_{0}^{\phi}\sqrt{1-k^2\sin^{2}\theta}{\rm d} \theta$ is the incomplete elliptic integral of the second kind, and $K(k)=\int_{0}^{\frac{\pi}{2}} \frac{1}{\sqrt{1-k^2\sin^{2}\theta}}{\rm d} \theta$ is the complete elliptic integral of the first kind \cite{abramowitz1964handbook}. 
With some mathematical simplification, the results in \eqref{eq:areaOfCassiniOval} is achieved. This completes the proof.

\section{}
\label{append:lemma1}

In this appendix, we provide a proof of Lemma \ref{lem:ellipsePCDF}.
The CDF can be derived using the area ratio, and then PDF is achieved by taking derivative of the CDF with the aid of $\frac{{\rm d} }{{\rm d} k}E(k)=\frac{E(k)-K(k)}{k}$, and $\frac{{\rm d} }{{\rm d} k}K(k)=\frac{E(k)}{k(1-k^2)}-\frac{K(k)}{k}$ \cite{abramowitz1964handbook}. This completes the proof.

\section{}\label{append:localOp}
In this appendix, we provide a proof of Lemma \ref{lem:localOp}. For product-distance path loss law, given the $d_{\rm BIU}$, the coordinates of the optimal location of the hybrid passive and active IRS satisfies

%\begin{small}
\begin{equation}\label{eq:coord_1}
	(x^2+y^2)^2-2c^2(x^2-y^2)=d_{\rm BIU}^{2}-c^{4},
\end{equation} 
%\end{small}

%\begin{small}
\begin{equation}\label{eq:coord_2}
	(x+c)^{2}+y^{2}=\left(\frac{P_{\rm T}\delta^{2}}{P_{\rm F}\delta_{\rm F}^{2}}\right)^{\frac{1}{\alpha}}d_{\rm BIU}.
\end{equation} 
%\end{small}
Thus, we can achieve that 
%\begin{small}
\begin{equation}\label{eq:cood_y}
	y^2 = \left(\frac{P_{\rm T}\delta^{2}}{P_{\rm F}\delta_{\rm F}^{2}}\right)^{\frac{1}{\alpha}}d_{\rm BIU} - (x+c)^2. 
\end{equation}
%\end{small}
Next, substitute \eqref{eq:cood_y} into \eqref{eq:coord_1} and \eqref{eq:coord_2}, \eqref{eq:cood} can be achieved with some mathematical simplifications. Following similar procedure, the results for sum-distance path loss law can be achieved. This completes the proof.

\section{} \label{append:mixNaka}
In this appendix, we provide a proof for Proposition~\ref{prop:channelModel}. As the channel of each individual link, e.g., BS$\rightarrow$IRS or IRS$\rightarrow$UE, follows generalized fading distribution, the channel gain of the cascaded channel can be modeled as a mixture Gamma distribution as proved in \cite{10458985}. Then, the PDF of mixture Nakagami-$m$ distribution for channel coefficient of the cascaded link (BS$\rightarrow$IRS$\rightarrow$UE) in \eqref{append:mixtureNaka} can be achieved with the assistance of the PDF scaling law $f_{\rm X}(x)=2xf_{\rm S}(x^2)$, where $x=\sqrt{s}$. This completes the proof.

\section{} \label{append:CLT_h}
In this appendix, we provide a proof for the approximated channel modeling for the mixture channel of $L$ IRSs aided mixture channel in Lemma~\ref{lem:coefficient}. As the channel coefficient of each cascaded link follows mixture Nakagami-$m$ distribution, the average channel coefficient of the mixture Nakagami-$m$ distributed link is derived as follows: 
%\begin{small}
\begin{equation}
	\begin{split}
		& \mathbb{E}[X] = \sum_{i=1}^{I} \varepsilon_{i} \int_{0}^{\infty} x^{2\beta_{i}-1}e^{-\xi_{i}x^2}\cdot 2x{\rm d} x \\& \overset{(a)}{=} \sum_{i=1}^{I} \varepsilon_{i} \int_{0}^{\infty} t^{\beta_{i}-\frac{1}{2}}e^{-\xi_{i}t} {\rm d}t  \overset{(b)}{=}  \sum_{i=1}^{I} \varepsilon_{i}\left( \beta_{i}-\frac{1}{2}\right) ! \xi_{i}^{-\beta_{i}-\frac{1}{2}}, \notag
	\end{split}
\end{equation}
%\end{small}
where step (a) is achieved by a substitution of $t=x^2$, and step (b) is derived with the assistance of \cite[eq.3.351.5]{abramowitz1964handbook}. As phase alignment is employed, the mean of the channel coefficient is %given by
\begin{equation}
	\mu = \mathbb{E}[|{h}_{\rm S}|] = \mathbb{E}[|{ h}_{{\rm BIU},1}|] + \cdots + \mathbb{E}[|{ h}_{{\rm BIU},L}|] = \sum_{l=1}^{L} \mu_{l}, 
\end{equation} 
where $\mu_{l}=\mathbb{E}[|{h}_{{\rm BIU},l}|]$, and its variance is given by
\begin{equation}
	\delta_{\rm S}^{2} = \mathbb{E}[|{ h}_{\rm S}|^{2}]  - \mu ^{2} ,
\end{equation}
with
%\begin{small}
\begin{equation}
	\mathbb{E}[|{ h}_{\rm S}|^{2}] =   \mathbb{E} \left[ \sum_{l_{1}=1}^{L}\sum_{l_{2}=1}^{L} {h}_{l_1} {h}_{l_2} \right] =  L+ \sum_{l_{1}=1}^{L}\sum_{l_{2}=1, l_{2}\neq l_{1}}^{L}  \mu_{l_{1}}\mu_{l_{2}} , \notag
\end{equation}
where ${h}_{l_1}$ and ${h}_{l_2}$ represent the channel of ${l_1}$-th and ${l_2}$-th BS$\rightarrow$IRS$\rightarrow$UE link, respectively.
%\end{small} 
This completes the proof.

\bibliographystyle{IEEEtran}  
\bibliography{references}

			\begin{IEEEbiography}[{\includegraphics[width=1in,height=1.25in,clip,keepaspectratio]{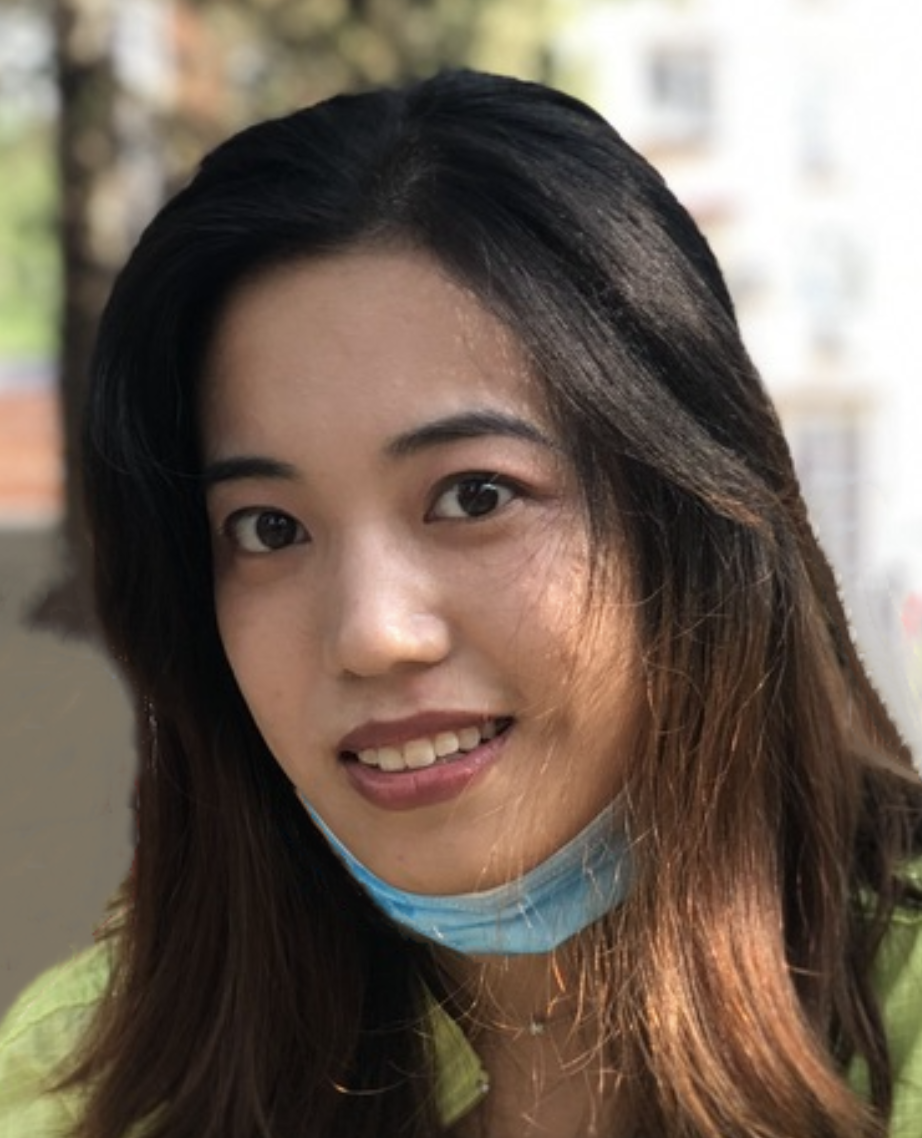}}]{Yunli Li} received the B.S. degree in electronic information engineering from Nanjing University of Aeronautics and Astronautics, in 2018; the M.S. degree in telecommunications from Hong Kong University of Science and Technology, in 2019. She is currently pursuing the Ph.D. degree with the Department of Electronic Engineering from the City University of Hong Kong. Her research interests include intelligent reflecting surface, channel modeling, and stochastic geometry.
				
			\end{IEEEbiography}

			\begin{IEEEbiography}[{\includegraphics[width=1in,height=1.25in, clip,keepaspectratio]{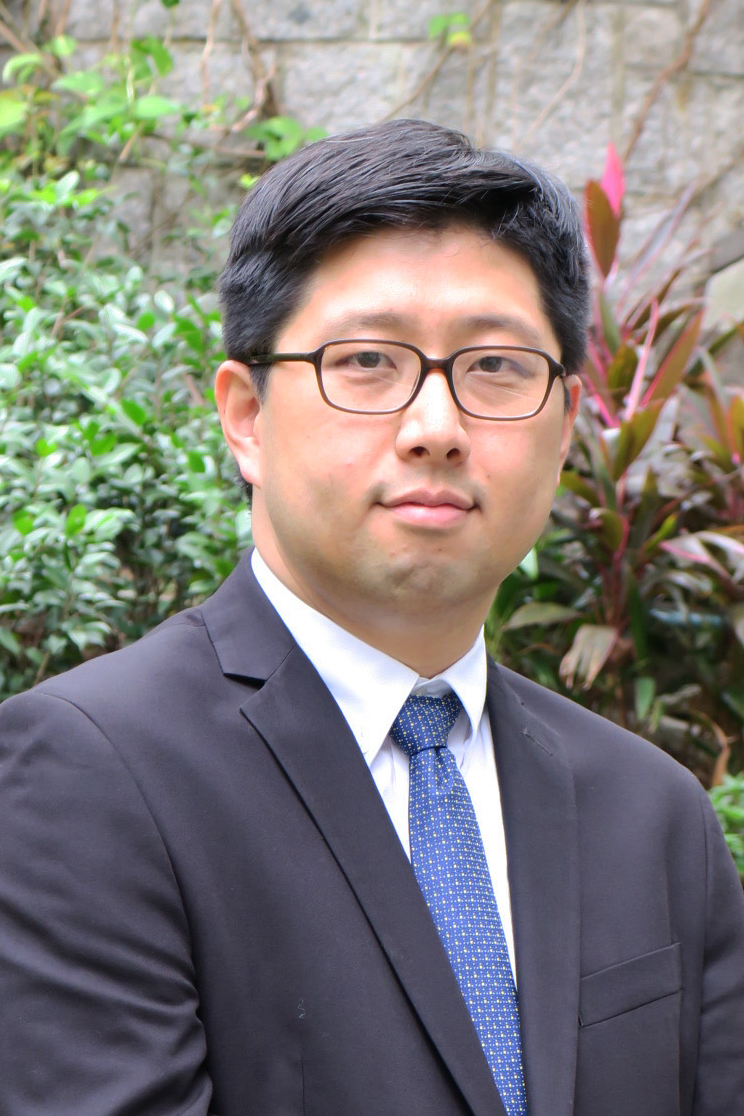}}]{Young Jin Chun} received the B.S. degree from Yonsei University in 2004, the M.S. degree from the University of Michigan in 2007, and the Ph.D. degree from Iowa State University in 2011, all in electrical engineering. From September 2011 to October 2018, he had been with Sungkyunkwan University, Qatar University, and Queen's University Belfast in various academic positions. Dr. Chun joined the City University of Hong Kong in November 2018 as an Assistant Professor in Wireless Communication in the Department of Electronic Engineering. His research interests are primarily in the area of wireless communications with an emphasis on stochastic geometry, system-level network analysis, device-to-device communications, and various use-cases of 5G communications.
			\end{IEEEbiography}

\end{document}